%% file: tr.tex
\documentclass[twoside,leqno,twocolumn]{article}
\usepackage{ltexpprt}

\newif\ifDoubleBlind
\DoubleBlindfalse

\newif\ifTR
\TRtrue
\usepackage{booktabs} 
\usepackage{numprint}
\usepackage{mathtools}
\usepackage{multirow}
\usepackage{pdflscape}
\usepackage{afterpage}
\usepackage{graphics}
\usepackage{placeins} 
\input{makros.tex}

\usepackage{amsmath}
\usepackage{amsthm}
\newtheorem{Satz}{Theorem}[section]
\numberwithin{equation}{section}

\usepackage[ruled,vlined,linesnumbered,norelsize]{algorithm2e}
\SetKw{KwAs}{as}
\SetKw{KwNew}{new}

\usepackage{hyperref}
\usepackage{wrapfig}
\makeatletter
\def\maxwidth{ %
  \ifdim\Gin@nat@width>\linewidth
    \linewidth
  \else
    \Gin@nat@width 
  \fi
}
\makeatother

\usepackage{alltt}
\newcommand{\ie}{i.e.~}
\newcommand{\etal}{et~al.\ }

\usepackage[square,numbers,sort&compress]{natbib}

\begin{document}

\ifTR
\title{\Large Multilevel Acyclic Hypergraph Partitioning\thanks{Partially supported by DFG grants DFG SA 933/11-1 and SCHU 2567/1-2.}}
\fi

\author{Merten Popp\thanks{GrAI Matter Labs.}
\and Sebastian Schlag\thanks{Karlsruhe Institute of Technology.}
\and Christian Schulz\thanks{Heidelberg University.}
\and Daniel Seemaier\thanks{Karlsruhe Institute of Technology.}}
\date{}
\maketitle

\fancyfoot[C]{\thepage}

\begin{abstract}
A directed acyclic hypergraph is a generalized concept of a directed acyclic graph, where each hyperedge can contain an arbitrary number of tails and heads.  Directed hypergraphs can be used to model data flow and execution dependencies in streaming applications.  Thus, hypergraph partitioning algorithms can be used to obtain efficient parallelizations for multiprocessor architectures.  However, an acyclicity constraint on the partition is necessary when mapping streaming applications to embedded multiprocessors due to resource restrictions on this type of hardware.  The acyclic hypergraph partitioning problem is to partition the hypernodes of a directed acyclic hypergraph into a given number of blocks of roughly equal size such that the partition is acyclic while minimizing an objective function. 

Here, we contribute the first $n$-level algorithm for the acyclic hypergraph partitioning problem.  
Based on this, we engineer a memetic algorithm to further reduce communication cost, as well as to improve scheduling makespan on embedded multiprocessor architectures. Experiments indicate that our algorithm outperforms previous algorithms that focus on the directed acyclic graph case which have previously been employed in the application domain.  Moreover, our experiments indicate that using the directed hypergraph model for this type of application yields a significantly smaller makespan.  
\end{abstract}

\clearpage 
\pagestyle{fancy}

\section{Practical Motivation}
\label{intro}

This research is inspired by computer vision and imaging applications which typically have a high demand for computational power.
Quite often, these applications run on embedded devices that have limited computing resources and also a tight thermal
budget. This requires the use of specialized hardware and a programming model
that allows to fully utilize the computing resources for streaming applications.
Directed graphs can be used to model data flow and execution dependencies in streaming applications.
Thus, graph partitioning algorithms can be used to obtain efficient parallelizations
for multiprocessor architectures.
However, when mapping streaming applications to embedded multiprocessors, memory-size restrictions on this type of hardware require the partitioning to be acyclic.
The acyclic graph partitioning problem is NP-complete~\cite{DBLP:conf/wea/MoreiraPS17}, and there is no constant factor approximation~\cite{DBLP:conf/wea/MoreiraPS17}.
Hence, heuristic algorithms are used in practice. Very recently, several multilevel and memetic algorithms have been proposed for directed acyclic graphs~(DAGs)~\cite{DBLP:conf/wea/MoreiraPS17,DBLP:conf/gecco/MoreiraP018,DBLP:journals/siamsc/HerrmannOUKC19}.

Here, we generalize the partitioning problem to directed acyclic \emph{hypergraphs}. This enables us to use more realistic objective functions. 
To be more precise, a directed acyclic hypergraph is a generalized concept of a directed acyclic graph, where each hyperedge can contain an arbitrary number of tails and heads.
Our focus is on acyclic hypergraphs where hyperedges can have one head and arbitrary many tails -- however the algorithms can be easily extended to the more general case.
Hypergraphs, as opposed to regular graphs, allow application developers to model dataflow of data streams with multiple producers/consumers more precisely.
The acyclic hypergraph partitioning problem is to partition the hypernodes of a directed acyclic hypergraph into a given number of blocks of roughly equal size such that the corresponding quotient graph is acyclic while minimizing an objective function on the partition. 
As the quality of this partitioning has a strong impact on performance and partitions often only have to be computed once to be used many times, our focus in this work is on \emph{solution quality}.

A highly successful heuristic to partition large undirected hypergraphs is the \emph{multilevel} approach~\cite{SPPGPOverviewPaper}.
Here, the hypergraph is recursively \emph{contracted} to obtain smaller hypergraphs which should reflect the same basic structure as the input. After applying an \emph{initial partitioning} algorithm to the smallest hypergraph, contraction is undone and, at each level, a
\emph{local search} method is used to improve the partitioning induced by the coarser level. 
The intuition behind this approach is that a good partition at one level of the hierarchy will also be a good partition on the next finer level. Hence, 
local search algorithms are able to explore local solution spaces very effectively. 
However, local search algorithms often get stuck in local optima~\cite{hMetisKway}. 

While multiple independent repetitions of the multilevel algorithm can help to improve the result, even a large number of repeated executions can only scratch the surface of the huge space of possible partitionings. 
In order to explore the global solution space extensively we need more sophisticated metaheuristics. 
This is where memetic algorithms (MAs), i.e., genetic algorithms combined with local search~\cite{KimHKM11}, 
come into play. Memetic algorithms allow for effective exploration (global search) and exploitation (local search) of the solution space.

We have three main contributions. Firstly, we contribute the first $n$-level algorithm for the \emph{acyclic} hypergraph partitioning problem. Based on this, we engineer a memetic algorithm to further reduce communication cost. Experiments indicate that our algorithms scale well to large instances and compute high-quality acyclic hypergraph partitions. Moreover, our algorithms outperform previous algorithms that partition directed acyclic graphs, which have been the model previously employed by users in the application domain. Lastly, our experiments indicate that using the directed hypergraph model for this type of application has a significant advantage over the directed graph model in practice.

\section{Preliminaries}
\textbf{Notation and Definitions.}
An \textit{undirected hypergraph} $H=(V,E,c,$ $\omega)$ is defined as a set of $n$ hypernodes/vertices $V$ and a
set of $m$ hyperedges/nets $E$ with vertex weights $c:V \rightarrow \mathbb{R}_{>0}$ and net 
weights $\omega:E \rightarrow \mathbb{R}_{>0}$, where each net $e$ is a subset of the vertex set $V$ (i.e., $e \subseteq V$). The vertices of a net are called \emph{pins}.
We extend $c$ and $\omega$ to sets, i.e., $c(U) :=\sum_{v\in U} c(v)$ and $\omega(F) :=\sum_{e \in F} \omega(e)$.
A vertex $v$ is \textit{incident} to a net $e$ if $v \in e$. $\mathrm{I}(v)$ denotes the set of all incident nets of $v$. 
The set $\Gamma(v) := \{ u~|~\exists~e \in E : \{v,u\} \subseteq e\}$ denotes the neighbors of $v$.
The \textit{size} $|e|$ of a net $e$ is the number of its pins. 

A \emph{$k$-way partition} $\Pi$ of an undirected hypergraph $H$ is a partition of its vertex set into $k$ \emph{blocks} $\mathrm{\Pi} = \{V_1, \dots, V_k\}$ 
such that $\bigcup_{i=1}^k V_i = V$, $V_i \neq \emptyset $ for $1 \leq i \leq k$ and $V_i \cap V_j = \emptyset$ for $i \neq j$.
We use $b[v]$ to refer to the block of vertex $v$.
We call a $k$-way partition $\mathrm{\Pi}$ \emph{$\mathrm{\varepsilon}$-balanced} if each block $V_i \in \mathrm{\Pi}$ satisfies the \emph{balance constraint}:
$c(V_i) \leq L_{\max} := (1+\varepsilon)\lceil \frac{c(V)}{k} \rceil$ for some parameter $\mathrm{\varepsilon}$. 
Given a $k$-way partition $\mathrm{\Pi}$, the number of pins of a net $e$ in block $V_i$ is defined as
$\mathrm{\Phi}(e,V_i) := |\{v \in V_i~|~v \in e \}|$. 
For each net $e$, $\mathrm{\Lambda}(e) := \{V_i~|~ \mathrm{\Phi}(e, V_i) > 0\}$ denotes the \emph{connectivity set} of $e$.
The \emph{connectivity} of a net $e$ is the cardinality of its connectivity set: $\mathrm{\lambda}(e) := |\mathrm{\Lambda}(e)|$.
A net is called \emph{cut net} if $\mathrm{\lambda}(e) > 1$.

The generalized version of undirected hypergraphs are \emph{directed hypergraphs}. 
A directed hypergraph is an undirected hypergraph where each hyperedge $e \in E$ is divided into a set of tails $e^T \subseteq e$ and heads $e^H \subseteq e$ that fulfill $e^T \cup e^H = e$ and $e^T \cap e^H = \emptyset$.
In a directed hypergraph, a cycle $C$ of length $k$ is a sequence of hypernodes $C = (v_1, \dots, v_k, v_{k + 1} = v_1)$ such that for every $i = 1, \dots, k$, there exists some hyperedge $e \in E$ with $v_i \in e^T$ and $v_{i + 1} \in e^H$. 
Furthermore, we require that $v_i \neq v_j$ for $i \neq j$, $1 \le i, j \le k$.
Each hypernode has predecessors $\Gamma^{-}(u) \coloneqq \{ v \mid v \in e^T, u \in e^H \text{ for some } e \in E \}$ and successors $\Gamma^{+}(u) \coloneqq \{ v \mid u \in e^T, v \in e^H \text{ for some } e \in E \}$. 
We refer to directed hypergraphs that do not contain any cycles as \emph{directed acyclic hypergraphs} (DAHs).
The \emph{quotient graph} for a partitioned directed acyclic hypergraph $H$ contains a node $v_i$ for each block $V_i$ and an edge $(v_i, v_j)$ if $H$ contains a
 hyperedge $e$ with tail pins in $V_i$ and head pins in $V_j$, i.e., $e^T \cap V_i \neq \emptyset$ and $e^H \cap V_j \neq \emptyset$. 

	Let $H = (V, E)$ be a DAH. The \emph{toplevel} of a node $v \in V$, denoted by $\texttt{top}[v]$, is the length of the longest path from any node with indegree zero in $H$ to $v$. 
	In particular, nodes $s$ with indegree zero have toplevel~$\texttt{top}[s] = 0$.
	Let $\mathcal{C} = \{\mathcal{C}_1, \dots, \mathcal{C}_k\}$ be a clustering of $V$ such that for each $\mathcal{C}_i \in \mathcal{C}$ and for all $u, v \in \mathcal{C}_i$, $\lvert \texttt{top}[u] - \texttt{top}[v] \rvert \le 1$.
	We refer to clusters $\mathcal{C}_i$ where all $u, v \in \mathcal{C}_i$ have $\texttt{top}[u] = \texttt{top}[v]$ as \emph{single-level} clusters and
        to clusters $\mathcal{C}_j$ that contain at least one pair of nodes $u, v \in \mathcal{C}_j$ with $\lvert \texttt{top}[u] - \texttt{top}[v] \rvert = 1$ as \emph{mixed-level} clusters.

The \emph{$k$-way hypergraph partitioning problem} is to find an $\varepsilon$-balanced $k$-way partition $\mathrm{\Pi}$ of a hypergraph $H$ that
minimizes an objective function over the cut nets for some $\varepsilon$.
The most commonly used cost functions are the \emph{cut-net} metric $\text{cut}(\mathrm{\Pi}) := \sum_{e \in E'} \omega(e)$ and the
\emph{connectivity} metric $(\mathrm{\lambda} - 1)(\mathrm{\Pi}) := \sum_{e\in E'} (\mathrm{\lambda}(e) -1)~\omega(e)$, where $E'$ is the set of all cut nets~\cite{UMPa,donath1988logic}.
Optimizing either of both objective functions is known to be NP-hard \cite{Lengauer:1990}.
In this paper, we use the connectivity-metric $(\mathrm{\lambda} - 1)(\mathrm{\Pi})$.
The \emph{hypergraph partitioning problem for directed acyclic hypergraphs} is the same as before, but with the further restriction that the resulting quotient graph must also be \emph{acyclic}. 
\emph{Contracting} a pair of vertices $(u, v)$ means merging $v$ into $u$.
The weight of $u$ becomes $c(u) := c(u) + c(v)$. We connect $u$ to the former neighbors $\Gamma(v)$ of $v$ by replacing 
$v$ with $u$ in all nets $e \in \mathrm{I}(v) \setminus \mathrm{I}(u)$ and remove $v$ from all nets $e \in \mathrm{I}(u) \cap \mathrm{I}(v)$.
\emph{Uncontracting} a vertex $u$ reverses the contraction.\\


\noindent\textbf{Related Work}.
\label{s:related}
 Driven by applications in VLSI design and scientific computing, hypergraph partitioning (HGP) has evolved into a broad research area since the 1960s.
We refer to existing literature~\cite{Alpert19951,Papa2007,trifunovic2006parallel,SPPGPOverviewPaper,DBLP:reference/bdt/0003S19,SchlagHGP} for an extensive overview.
In the following, we focus on issues closely related to the contributions of our paper.
Well-known multilevel HGP software packages with certain distinguishing characteristics include \cite{PaToH,hMetisRB,hMetisKway,ahss2017alenex,KaHyPar-MF-JEA,hs2017sea,KaHyPar-R,Mondriaan,MLPart,Zoltan,Parkway2.0,SHP,DBLP:conf/dimacs/CatalyurekDKU12,Aykanat:2008}.
Parallel algorithms for the graph partitioning~\cite{DBLP:journals/jpdc/KarypisK98,DBLP:conf/alenex/Schlag0SS19,DBLP:journals/tpds/MeyerhenkeSS17,DBLP:conf/ipps/SlotaRDM17,DBLP:conf/europar/Akhremtsev0018,DBLP:conf/ipps/LasalleK13} and hypergraph partitioning problem~\cite{DBLP:journals/jpdc/TrifunovicK08,DBLP:conf/ipps/DevineBHBC06} are also available. \\

\noindent\textit{Evolutionary Partitioning/Clustering.}
Memetic algorithms~(MAs) were introduced by Moscato~\cite{MAs} and formalized by Radcliffe and Surry~\cite{RadcliffeS94} as an extension to the concept of genetic algorithms~(GAs)~\cite{Holland:1975}.
While GAs effectively explore the \emph{global} solution space, MAs additionally allow for exploitation of the \emph{local} solution space by incorporating local search
methods into the genetic framework. We refer to the work of Moscato and Cotta~\cite{Moscato2010} for an introduction.
There is a wide range of evolutionary/memetic algorithms for the undirected hypergraph partitioning problem~\cite{SaabR89,Hulin1991,BuiMoon94,HuMoerder85,Areibi00anintegrated,HypergraphKFM,AreibiY04,FeoRS94,ArmstrongGAD10,KimM04,Kim2004}.
Recently, an memetic multilevel algorithm for undirected hypergraph partitioning  has been proposed~\cite{DBLP:conf/gecco/AndreS018}.
 Note that except Ref.~\cite{DBLP:conf/gecco/AndreS018}, \emph{none} of the above algorithms makes use of the multilevel~paradigm. 
We refer to the survey of Kim \etal \cite{KimHKM11} for an overview and more material on genetic approaches for graph partitioning. Recent approaches for graph partitioning include~\cite{soper2004combined,benlichao2010,ChardaireBM07,kaffpaE}.

Recently, multilevel algorithms for DAG partitioning have been proposed \cite{DBLP:conf/gecco/MoreiraP018,DBLP:journals/siamsc/HerrmannOUKC19,DBLP:conf/wea/MoreiraPS17}.
In Ref.~\cite{DBLP:conf/wea/MoreiraPS17}, a memetic algorithm with natural combine operations provided by a multilevel framework for the DAG partitioning problem can be found.
Here, the objective of the evolutionary algorithm is also modified to serve a specific objective function that is needed in a real-world imaging application.
To the best of our knowledge, there is currently \emph{no} multilevel or memetic algorithm for directed hypergraph partitioning. \\

\noindent\textit{Hypergraph Partitioning using KaHyPar.}
Since our algorithms are build on top of the KaHyPar \cite{KaHyPar-R} framework, we briefly review its core components.
     KaHyPar instantiates the multilevel paradigm in the
extreme $n$-level version, removing only a \emph{single} vertex between two levels.
Furthermore, it incorporates global information about the structure of the hypergraph into the coarsening process
by using community detection in a preprocessing step and preventing inter-community contractions during coarsening.
Vertex pairs $(u,v)$ to be contracted are determined using the heavy-edge rating function $r(u,v) := \sum_{e \in E'}  \omega(e)/(|e| - 1)$, where $E' := \{\mathrm{I}(u) \cap \mathrm{I}(v)\}$.

After coarsening, a portfolio of simple algorithms is used to create an initial partition of the coarsest hypergraph. During uncoarsening,
strong localized local search heuristics based on the FM algorithm~\cite{FMAlgorithm,HypergraphKFM} are used to refine
the solution by moving vertices to other blocks in the order of improvements in the optimization objective.
Recently, KaHyPar was extended with a refinement algorithm based on maximum-flow computations~\cite{KaHyPar-MF-JEA}.
Unless mentioned otherwise, we use the default configurations provided by~the~authors~\cite{configurationKaHyPar}. 

\section{Multilevel Memetic DAH Partitioning}
We now explain our core contribution and present the first algorithm for the 
problem of computing acyclic partitions of directed acyclic hypergraphs with all its components.\\

\noindent\textbf{Overview.} Our hypergraph partitioner with acyclicity constraints is based on KaHyPar.
To cope with the acyclicity constraint, we extend the scope of KaHyPar to include directed hypergraphs
and implement new algorithms for coarsening, initial partitioning, and refinement of directed hypergraphs and acyclic partitions. 

Similar to previous work on DAG partitioning~\cite{DBLP:conf/gecco/MoreiraP018,DBLP:journals/siamsc/HerrmannOUKC19,DBLP:conf/wea/MoreiraPS17},
we \emph{first} compute an initial solution for the DAH by moving the initial partitioning phase \emph{before} the coarsening phase.
This initial solution is then used during coarsening to prevent contractions that would lead to cycles in the quotient graph. This is done by
only selecting pairs of vertices for contraction that were placed in the same block in the initial partition.
Once the coarsening algorithm terminates, the contraction operations are undone in reverse order during the uncoarsening phase. 
After each uncontraction operation, we use a localized refinement algorithm to improve the current solution.
Note that in contrast to the prominent multilevel partitioning scheme that has roughly $O(\log n)$ levels,
the refinement algorithm is employed after \emph{every single} uncontraction operation.

To compute a $k$-way partition, our algorithm performs recursive bipartitioning.
We first compute a balanced bipartition of the input DAH by computing an initial
acyclic bipartition, which is then coarsened and refined using the techniques
described in Section~\ref{ssec:coarsening} and Section~\ref{ssec:refinement}.
Afterwards, we build an induced acyclic subhypergraph for each block and recursively bipartition
the blocks until the desired number of blocks is reached. 
Lastly, we perform a V-cycle using our $k$-way FM algorithm described in Section~\ref{ssec:refinement} to further improve it.

\subsection{Initial Partitioning}
\label{ssec:ip}
This section describes our approaches for obtaining an initial partition of the directed acyclic hypergraph. 
Each algorithm starts with an unpartitioned directed acyclic hypergraph $H = (V, E)$ and produces a partition of $V$ into blocks $V_1, \dots, V_k$ for a fixed number of blocks $k$. \\

\noindent \textbf{Initial Partitioning via Topological Ordering.} For the DAG partitioning problem, Moreira \etal \cite{DBLP:conf/gecco/MoreiraP018,DBLP:conf/wea/MoreiraPS17} compute their initial partition based on a topological ordering of the graph. 
We implement the same approach for directed acyclic hypergraphs to obtain an initial $k$-way partition.
First, we calculate a topological ordering of the nodes of the hypergraph using Kahn's algorithm \cite{KahnsAlgorithm} adapted for directed hypergraphs, i.e.,
we repeatedly order and remove vertices from the hypergraph with indegree zero. 
Iterating over the topological ordering, our algorithm greedily assigns vertices to a block until it is full, i.e., its weight reaches $\lceil \frac{c(V)}{k} \rceil$.
The algorithm then moves on to the next block.
Note that this approach always produces a balanced initial partition for hypergraphs with unit node weights. 
For weighted hypergraphs, it can produce an imbalanced initial partition. 
In this case, the refinement step must balance the partition.  \\
\begin{algorithm}[b!]
\DontPrintSemicolon
	\KwData{DAH $H = (V, E)$.}
	\KwResult{Bipartition $V = V_1 \dot\cup V_2$.}
        \tcp{partition as undirected hypergraph}
	$(V_1, V_2) \coloneqq$ \texttt{HG$(H, 2)$} \\
        \tcp{break quotient graph edge $(V_1, V_2)$}
	$(V'_1, V'_2) \coloneqq$ \texttt{Ref$($Balance$($FixCyclic$(V_1, V_2)))$}  \\
        \tcp{break quotient graph edge $(V_2, V_1)$}
	$(V''_1, V''_2) \coloneqq$ \texttt{Ref$($Balance$($FixCyclic$(V_2, V_1)))$}  \\
	\tcp{select bipartition with lower connectivity} 

	\lIf{\texttt{KM1$(V'_1, V'_2)$} $\le$ \texttt{KM1$(V''_1, V''_2)$}}{ 
		\KwRet{$(V'_1, V'_2)$}
	}\lElse{
		\KwRet{$(V''_1, V''_2)$}
	}
	\caption{Initial partitioning algorithms that makes use of a preexisting hypergraph partitioner \texttt{HG($\cdot, \cdot$)} for undirected hypergraphs. \texttt{KM1($\cdot, \cdot$)} denotes the connectivity metric of the given bipartition.}	
	\label{alg:undir_ip}
\end{algorithm}

\noindent \textbf{Initial Partitioning via Undirected Partitioning.}
This algorithm is based on the initial partitioning scheme for the DAG partitioning algorithm presented by Herrmann \etal \cite{DBLP:journals/siamsc/HerrmannOUKC19}.
Algorithm~\ref{alg:undir_ip} gives an overview over our approach for partitioning a directed hypergraph into $k=2$ blocks.
Given the DAH, we first obtain a bipartition of the \emph{undirected version} of the hypergraph, which is then projected onto the original DAH. 
Since ignoring directions and dropping the acyclicity constraint may lead to a bipartition that contains cycles,
we run an algorithm that fixes the partition (Algorithm~\ref{alg:makebipartitionacyclic}) and afterwards try to
improve the balance and cut of the now acyclic solution. This in done twice, once removing the
quotient graph edge from $V_1$ to $V_2$ and once removing the reverse edge. 
We then repeat the process but move predecessors instead of successors in Algorithm~\ref{alg:makebipartitionacyclic}.
Finally, we select the bipartition with the lowest cut as solution.

The DAH is turned into an undirected hypergraph by merging the tails and heads of each hyperedge.
This hypergraph is then used as input for a standard hypergraph partitioner that minimizes the connectivity objective.
In our experiments, we use KaHyPar-MF~\cite{KaHyParMF} and PaToH~\cite{PaToH} for this task,
since KaHyPar-MF regularly finds partitions with the lowest connectivity objective out of all hypergraph partitioners,
while PaToH is the fastest partitioner~\cite{KaHyParMF}. In case of KaHyPar, we use the strongest configuration (i.e., \texttt{km1\_direct\_kway\_sea18.ini}). In case of PaToH,
we use the default configuration.

If the computed bipartition violates the acyclicity constraint after projecting it onto the original DAH, 
we use Algorithm~\ref{alg:makebipartitionacyclic} to make it acyclic.
Roughly speaking, we select one edge in the quotient graph that we want to remove and move hypernodes from one block to the other one accordingly.
To be more precise, denote the two blocks by $V_1$ and $V_2$ and assume that we want to remove the quotient graph edge from $V_1$ to $V_2$. 
We start a breath-first search at every hypernode in $V_1$ that has successors in $V_2$. 
The search only scans successors in $V_2$ and moves every node from $V_2$ to $V_1$. 
Once the search has completed, no hypernode in $V_1$ has successors in $V_2$ and therefore the quotient graph edge from $V_1$ to $V_2$ is removed. 

The resulting acyclic partition might become imbalanced due to the movements from one block to the other. 
To cope with this problem, we run an additional balancing and refinement step afterwards. 
This step simply moves hypernodes from the overloaded block to the underloaded block. 
Note that we cannot move arbitrary hypernodes while keeping the bipartition acyclic. 
If we have an acyclic bipartition with blocks $V_1$ and $V_2$, and a quotient graph edge from $V_1$ to $V_2$, we can only move hypernodes in $V_1$ that have no successors in $V_1$. 
In an effort the keep the connectivity of the bipartition low, we sort the movable hypernodes in the overloaded block by their gain value, i.e.,
the reduction in the objective function if the node is moved, using a priority queue, and run our $2$-way FM refinement algorithm described in Section~\ref{ssec:refinement} afterwards.

\begin{algorithm}[b!]
	\KwData{Cyclic bipartition $(V_1, V_2)$ of DAH $H = (V, E)$.}	
	\KwResult{Acyclic bipartition.}
	$S \coloneqq$ \KwNew \texttt{Stack}$()$\; 
	\For{$u \in V_1$}{
		\If{$\Gamma^{+}(u) \cap V_2 \neq \emptyset$}{
			$S \coloneqq S \cup \{ u \}$\;
		}
	}
	\While{$S \neq \emptyset$}{
		$u \coloneqq S$.\texttt{pop}$()$\;
		\For{$v \in \Gamma^{+}(u) \cap V_2$}{
			$S \coloneqq S \cup \{ v \}$\;
			$V_1 \coloneqq V_1 \cup \{ v \}$\;
			$V_2 \coloneqq V_2 \setminus \{ v \}$\;
		}
	}
	\KwRet{$(V_1, V_2)$}
	\caption{Subroutine \texttt{FixCyclic}$(\cdot)$ referenced in Algorithm~\ref{alg:undir_ip}: moves nodes to make a bipartition acyclic.} 
	\label{alg:makebipartitionacyclic}
\end{algorithm}


\subsection{Acyclic Coarsening}
\label{ssec:coarsening}

Having computed an initial partition via one of the two techniques described in the previous section, we proceed to the \emph{acyclic} coarsening phase.
By restricting contractions to pairs of vertices that are in the same block, we prevent the quotient graph from becoming cyclic.
However, contractions within a block may still lead to cycles within the directed hypergraph.

Our approach therefore restricts the coarsening algorithm in KaHyPar to pairs of hypernodes that can be safely contracted while keeping the hypergraph acyclic. 
First, we compute a clustering in the DAH. This clustering can be safely contracted, i.e., the coarser directed hypergraph will not contain a cycle.
After we have computed the clustering, we perform cluster contraction by iteratively contracting pairs of hypernodes that are inside the same cluster (yielding an $n$-level algorithm). 
When working with hypergraphs where each hyperedge contains at most one head pin, we allow the algorithm to contract pairs of hypernodes where one hypernode is a head and the other one is a tail in the same hyperedge. 
In this case, the contracted hypernode becomes a head in the hyperedge. 
For hypergraphs that contain hyperedges with multiple heads, we restrict the hypernode pairs to such where both hypernodes have the same role in all shared hyperedges.
Our algorithm to compute clusterings is based on Theorem~\ref{thm:forbidden_edges_hyperdag}, which identifies pairs of hypernodes that should not be in the same cluster.
We start by restating Theorem 4.2 of Herrmann~\etal~\cite{DBLP:journals/siamsc/HerrmannOUKC19} in Theorem~\ref{thm:forbidden_edges}, which identifies forbidden edges in DAGs (edges that if contracted may create a cycle).
We extend this to DAHs in Theorem~\ref{thm:forbidden_edges_hyperdag}.

\begin{Satz}{(Forbidden Edges \cite{DBLP:journals/siamsc/HerrmannOUKC19}.)}\label{thm:forbidden_edges}
	Let $G = (V, E)$ be a DAG and $\mathcal{C} = \{ \mathcal{C}_1, \dots, \mathcal{C}_k \}$ be a clustering of $V$.
	If $\mathcal{C}$ is such that
		for any cluster $\mathcal{C}_i$ and for all $u, v \in \mathcal{C}_i$, it holds that $\lvert \texttt{top}[u] - \texttt{top}[v] \rvert \le 1$, \emph{and}
		for two different clusters $\mathcal{C}_i$ and $\mathcal{C}_j$, and for all $u \in \mathcal{C}_i$ and $v \in \mathcal{C}_j$, either $(u, v) \notin E$, or $\texttt{top}[u] \neq \texttt{top}[v] - 1$,
	then the coarser graph that results from simultaneously contracting all clusters is acyclic. 
\end{Satz}

\begin{Satz}(Forbidden Hyperedges)\label{thm:forbidden_edges_hyperdag}
	Let $H = (V, E)$ be a DAH and $\mathcal{C} = \{\mathcal{C}_1, \dots, \mathcal{C}_k\}$ be a clustering of $V$, such that 
		for any cluster $\mathcal{C}_i$ and for all $u, v \in \mathcal{C}_i$, $\lvert \texttt{top}[u] - \texttt{top}[v] \rvert \le 1$, \emph{and}
		for two different \textbf{mixed-level clusters} $\mathcal{C}_i$ and $\mathcal{C}_j$, and for all $u \in \mathcal{C}_i$ and $v \in \mathcal{C}_j$, 
		either there is no hyperedge $e \in E$ with $u \in e^T$ and $v \in e^H$, or $\lvert \texttt{top}[u] - \texttt{top}[v] \rvert > 1$,
	then the coarser DAH that results from simultaneously contracting all clusters is acyclic. 
\end{Satz}

\begin{proof}
First, note that given a directed hypergraph $H$, one can construct a directed graph $G$ that is equivalent to the hypergraph in regards to the acyclicity constraint by replacing each hyperedge $e$ by a complete bipartite graph from $e^T$ to $e^H$. 
$G$ is acyclic if and only if $H$ is acyclic and an acyclic partition of $G$ is also an acyclic partition of $H$ and vice-versa. 
        Assume that the coarser hypergraph contains a cycle. 
        With the equivalence above and Theorem~\ref{thm:forbidden_edges}, the cycle must contain at least one single-level cluster $C_i$.
        Moreover, since the nodes of $C_i$ have the same toplevel, the cycle must have length at least $2$.
        Let $C_{i - 1}$ be the predecessor and $C_{i + 1}$ be the successor of $C_i$ in the cycle. 
        Let $t$ be the lowest toplevel of nodes in $C_{i - 1}$. 
        Then, the toplevel of nodes in $C_i$ and $C_{i + 1}$ is at least $t + 1$, which forbids a path from $C_{i + 1}$ to $C_{i - 1}$, a contradiction. 
\end{proof}

Note that the difference between Theorem~\ref{thm:forbidden_edges} (Theorem 4.2~\cite{DBLP:journals/siamsc/HerrmannOUKC19}) and Theorem~\ref{thm:forbidden_edges_hyperdag} lies in the distinction between single-level and mixed-level clusters: the second condition must only hold for pairs of mixed-level clusters. 
Since adjacent nodes in DAGs always have different toplevels, the clustering algorithm by Herrmann \etal \cite{DBLP:journals/siamsc/HerrmannOUKC19} only produces mixed-level clusters. 
In contrast, DAHs might contain adjacent hypernodes with the same toplevel, justifying this~distinction. 

Based on this theorem, our clustering algorithm works as follows. 
At first, all hypernodes are in their own singleton cluster.
For each hypernode $u$ that is still in a singleton cluster, we rate each neighbor using the heavy-edge rating function already implemented in KaHyPar. 
We select the highest-rated neighbor $v$ whose cluster can include $u$ without violating the first condition from Theorem~\ref{thm:forbidden_edges_hyperdag}. 
If all hypernodes in $v$'s cluster have the same toplevel as $u$, we know that we can safely add $u$ to $v$'s cluster without inducing a cycle in the coarser hypergraph. 
Otherwise, we temporarily add $u$ to $v$'s cluster, search the hypergraph for edges violating the second condition in Theorem~\ref{thm:forbidden_edges_hyperdag},
and then check whether they induce a cycle in the contracted hypergraph. 
To be more precise, let the toplevel of hypernodes in $v$'s cluster be $t$ and $t + 1$. 
We maintain a queue of hypernodes that are to be processed. 
Initially, the queue contains all hypernodes in $v$'s cluster with toplevel $t$. 
For each hypernode $x$ in the queue, we examine its successors. 
If we find that a successor $y$ is in another cluster, we add all hypernodes from $y$'s cluster with toplevel $t$ to the queue. 
If $y$ is in $v$'s cluster, but $x$ is not, the search found a cycle in the coarsened DAH. 
At this point, we abort the search, remove $u$ from $v$'s cluster and move on to the next hypernode. 
If the search does not find a cycle, we leave $u$ in $v$'s cluster and move on.

After one round of clustering, we contract all hypernodes inside the same cluster pair by pair and start the next round of the algorithm on the resulting hypergraph. 
To introduce more diversity between rounds, we alternate between using toplevels and reversed toplevels, i.e., the maximum distance from a node to any node with outdegree zero. 
This process is repeated until the algorithm can no longer find any non-singleton clusters or the number of hypernodes drops below $160 k$. The second condition is the same stopping criterion used in \cite{KaHyPar}.

\subsection{Acyclic Refinement}
\label{ssec:refinement}
During uncoarsening, contraction operations are undone and after each uncontraction, we execute a localized refinement algorithm to improve the solution.
As our algorithm recursively bipartitions the input hypergraph, we use a $2$-way local search algorithm that 
improves the objective function by exchanging nodes between two blocks.
The memetic algorithm discussed in the next sections additionally employs a $k$-way local search algorithm to improve individuals.
In the following, we define the \emph{gain} of a node move as its resulting reduction in the objective function.\\

\noindent\textbf{$2$-way FM Refinement.} 
In the $2$-way setting, we use a variation of the well-known FM algorithm~\cite{FMAlgorithm} to improve the partition.
This algorithm moves hypernodes with the highest gain between the two blocks, while making sure to only consider movements that keep the partition acyclic. 
Over the course of the algorithm, it keeps track of the best bipartition. 
Once a stopping criterion decides that the refinement is unlikely to find a further improvement of the bipartition, it rolls back to the best partition found. 

More precisely, the algorithm uses two priority queues (one for each block) to keep track of hypernodes and their gains.
Each priority queue contains movable hypernodes in the corresponding block and their respective gain.
A hypernode is movable if and only if the move does not violate the balance constraint and if it can be moved to the other block without causing the partition to become cyclic. 
During $2$-way refinement, this is easy to decide: Let $V_1$ and $V_2$ denote the blocks of the bipartition and assume that $V_2$ is the successor of $V_1$ in the quotient graph. 
Then, a hypernode in $V_1$ can be moved to $V_2$ if and only if it does not have any successors in $V_1$. 
Analogously, a hypernode in $V_2$ can be moved to $V_1$ if and only if it does not have any predecessors in $V_2$. 
Therefore, it is sufficient to keep track of the number of successors or predecessors that a hypernode has in the same block. 
We implement this using a simple array that we compute once at the start of the uncoarsening phase and then update appropriately after every uncontraction or vertex movement.
In particular, we can use this counter to decide whether new hypernodes become movable (counter becomes zero) or unmovable (counter becomes nonzero) after a move. 
We then insert/remove those hypernodes into/from the appropriate priority queue. 

Initially, both priority queues are empty. 
After uncontracting a hypernode, the resulting hypernodes and their partners are inserted into the priority queues if they are movable. 
If no vertex is movable, the refinement step is skipped and the next hypernode is uncontracted. 
Otherwise, the algorithm pulls the hypernode with the highest gain value from the priority queue and marks it as visited.
Visited hypernodes are excluded from the remainder of the pass. If the move does not violate
the balance and acyclicity constraints, the vertex is then moved to the opposite block. Afterwards,
all unmarked movable neighbors of the moved vertex are inserted into their corresponding priority queue
and the gain values of all affected hypernodes are updated. 
This process continues until the stopping criterion decides that further improvements are unlikely or both priority queues become empty. The algorithm then reverts to the best partition found. 
We stop local search after \numprint{350} subsequent moves that did not lead to an improved cut.\\

\noindent\textbf{$k$-way FM Refinement.}
The $k$-way FM refinement aims to improve a given $k$-way partition and is based on the $k$-way FM refinement algorithm implemented in KaHyPar \cite{KaHyPar}. 
The algorithm maintains $k$ priority queues, one queue for each block. 
Each queue holds hypernodes that can be moved to the block, with the priority being the gain value of the respective move. 
We limit the set of movable hypernodes to border hypernodes and only consider moving a hypernode to adjacent blocks. 
The algorithm always performs the best move across all priority queues and after the stopping criterion is reached, the best found partition during the process is restored. 
Here, we use the same stopping criterion as in KaHyPar~\cite{ahss2017alenex} which stops local search as soon as further improvements are unlikely using a statistical model.
In terms of neighborhood, we implement the global moves neighborhood from Moreira et al.~\cite{DBLP:conf/wea/MoreiraPS17} (with a natural extension to DAHs).
This neighborhood was among the best in their experiments and also had the largest local search neighborhood among all neighborhoods considered in the paper.
Starting from a given partition of the DAH, the algorithm computes the quotient graph. 
While the local search algorithm moves nodes between the blocks, the quotient graph is kept up-to-date.
After moving a node, we check whether this move created a new edge in the quotient graph. 
In this case, we check it for acyclicity using Kahn's algorithm~\cite{KahnsAlgorithm} and undo the last
movement if it created a cycle.

\subsection{Memetic DAH Partitioning}\label{ssec:ma}
Multilevel algorithms can be extended in a natural way to obtain memetic algorithms~\cite{DBLP:conf/wea/BiedermannH0S18,DBLP:conf/gis/AhujaBS0W15,kaffpaE,DBLP:conf/wea/Sanders016,DBLP:conf/gecco/00010SW17}.
Schlag~\etal~\cite{DBLP:conf/gecco/AndreS018} propose a multilevel memetic algorithm for the undirected hypergraph partitioning problem.
We adapt this algorithm for directed hypergraph partitioning by exchanging the algorithms used in the recombination and mutation operators for undirected hypergraphs
with the multilevel algorithms for the DAH partitioning problem described above.
From a meta-optimization perspective, the overall structure of the algorithm by Schlag~\etal~\cite{DBLP:conf/gecco/AndreS018} remains unchanged.
Hence, we follow their description closely.
For the sake of completeness, we now present the overall structure of the algorithm and explain the recombination and mutation operations.

Given a hypergraph $H$ and a time limit $t$, the algorithm starts by creating an initial
population of $\mathcal{P}$ \emph{individuals}, which correspond to $\varepsilon$-balanced $k$-way partitions of $H$. This is done by running our DAH partitioning algorithm multiple times with different random seeds.
Note that all individuals in the population fulfill the acyclicity constraint due to the way they are created.
The size of the population $|\mathcal{P}|$ is determined adaptively by first measuring the time $t_\text{I}$ spent to create one individual.
Then, the population size is chosen such that the time to create $|\mathcal{P}|$ individuals is a certain fraction $\delta$ of the total running time budget $t$:
$|\mathcal{P}| := \max(3,\min(50,\delta\cdot(t/t_I)))$, where $\delta$ is a tuning parameter. This is the same as Schlag~\etal~\cite{DBLP:conf/gecco/AndreS018} used in their work.
The \emph{fitness} of an individual is the objective function that we optimize: the connectivity $(\lambda -1)(\mathrm{\Pi})$ of its partition $\mathrm{\Pi}$.
Note that if the input is a DAG instead of a DAH, then the fitness function is identical to the edge cut.
Our algorithm follows the \emph{steady-state} paradigm~\cite{EvoComp}, i.e., only \emph{one} offspring is created per generation.
To generate a new offspring, we use the multilevel recombination operator described below.
In order to sufficiently explore the global search space and to prevent premature convergence, we employ multilevel mutation operations. 

Each operator always creates a new offspring $o$.
We insert it into the population and evict the individual \emph{most similar} to the offspring among all individuals whose fitness is equal to or worse than $o$.
For each individual, we compute the multi-set $D :=  \{(e, m(e)) : e \in E \} $, where $m(e) := \lambda(e) - 1$ is the
multiplicity (i.e., number of occurrences) of $e$. Thus each cut net $e$ is represented $\lambda(e) - 1$ times in $D$.
The difference of two individuals $I_1$ and $I_2$ is then computed as $d(I_1, I_2) := | D_1 \ominus D_2 |$, where $\ominus$ is the symmetric difference.

\subsection{Recombination Operators}\label{sec:combine}
By generalizing the recombination operator framework of Schlag~\etal~\cite{DBLP:conf/gecco/AndreS018} from
undirected hypergraphs to directed hypergraphs, the two-point recombine operator described in this section assures that the fitness of the offspring is
\emph{at least as good as the best of both parents}. \\

\noindent\textbf{Two-Point Recombine.}
The operator starts with selecting parents for recombination using binary tournament selection (without replacement)~\cite{BlickleT96} w.r.t. the $(\lambda-1)$ objective.
A tournament size of two is chosen to keep the selection pressure low and to avoid premature convergence,
since all individuals already constitute high-quality solutions.
Both individuals/partitions are then used as input of a modified multilevel scheme as follows:

During coarsening, two vertices $u$ and $v$ are only allowed to be contracted if \emph{both parents agree
on the block assignment of both vertices}, i.e., if $b_1[u] = b_1[v] \wedge b_2[u] = b_2[v]$. Originally, this
is a generalization from multilevel memetic GP, \ie\cite{kaffpaE}, where \emph{edges running between two blocks are not eligible} for contraction
and therefore remain in the graph. 
In other words, the generalization allows two vertices of the same cut net to be contracted as long as the input individuals agree that they belong to the same block.
For directed HGP, this restriction ensures that cut nets $e$ remain in the coarsened hypergraph
\emph{and} maintain their connectivity $\lambda(e)$ regarding \emph{both} partitions. This modification is important for the optimization objective,
because it allows us to use the partition of the \emph{better} parent as initial partition of the offspring.
The stopping criterion during coarsening is changed such that it stops when no more contractions are possible. 
During uncoarsening, local search algorithms can then use this initial partitioning to (i) exchange good parts of the solution on the coarse
levels by moving few vertices and (ii) to find the best block assignment for those vertices, for which the parent partitions disagreed.
Since our local search algorithms guarantee non-decreasing solution quality, the final fitness of offspring solutions generated using
this kind of recombination is always \emph{at least as good as the better of both parents}.

\subsection{Mutation Operations and Diversification}\label{sec:mutate}
Our mutation operators are based on V-cycles~\cite{WalshawVcycle}.
The V-cycle technique reuses an already computed partition, e.g.~from a random individual $I$ of the population, as input for the multilevel approach and iterates
coarsening and local search phases several times using different seeds for randomization.
During coarsening, the quality of the solution is maintained by only contracting vertex pairs $(u,v)$
belonging to the same block ($b[u] = b[v]$). We define two different
mutation operators: one uses the current partition of the individual as initial partition of the coarsest hypergraph and guarantees non-decreasing
solution quality. The other one first generates a new partition using the multilevel scheme and then takes a random individual of the population and recombines it with the just created individual.
During initial partitioning the newly created parent partition is used as initial partition.
In both cases during uncoarsening, $k$-way local search algorithms improve the solution quality and thereby further mutate the individual. 
Since the second operator computes uses a new initial partition which might be different from the original partition of $I$, the fitness of offspring generated by this operator can be worse than the fitness of $I$.

\section{Experimental Evaluation}\label{sec:eval}
\textbf{System and Methodology.}\label{Methodology}
We implemented the multilevel and the memetic algorithm described in the previous section within the KaHyPar hypergraph partitioning framework. 
The code is written in C++ and compiled using g++-9.1 with flags \texttt{-O3} \texttt{-march=native}.
We plan to release the code.
All experiments are performed on one core of a machine that has two Intel Xeon E5-2670 Octa-Core (Sandy Bridge) processors 
clocked at $2.6$ GHz, $64$~GB main memory, $20$~MB L3-Cache and~8x256~KB~L2-Cache. 

Our experiments are structured as follows: as DAGs are a special type of DAHs, we start this section by comparing our algorithm with the state-of-the-art
for DAG partitioning by Herrmann~\etal~\cite{DBLP:journals/siamsc/HerrmannOUKC19} (HOUKC -- abbreviations of the author's last names) and Moreira~\etal~\cite{DBLP:conf/gecco/MoreiraP018}.
We then show how our algorithms perform against other competing DAH partitioning approaches.
Lastly, we perform experiments on a target platform and show that using DAH partitioning instead of DAG partitioning yields improved performance in practice.

We use \texttt{mlDHGP} to refer to our multilevel algorithm and \texttt{memDHGP} for our memetic algorithm.
For partitioning, we use $\epsilon = 0.03$ as imbalance factor and $k \in \{2, 4, 8, 16, 32\}$ for the number of blocks,
since those are the values used in previous work by Herrmann~\etal~\cite{DBLP:journals/siamsc/HerrmannOUKC19}.
As stated before, we use two undirected hypergraph partitioning algorithms (KaHyPar and PaToH) to compute initial partitions, and use subscript letters to indicate which algorithm we are currently using (\texttt{memDHGP$_K$} for KaHyPar and \texttt{memDHGP$_P$} for PaToH).\\

\noindent \textbf{Performance Profiles.}
In order to compare the solution quality of different algorithms, we use performance profiles \cite{DBLP:journals/mp/DolanM02}.
These plots relate the smallest minimum connectivity of all algorithms to the corresponding connectivity produced by each algorithm.
More precisely, the $y$-axis shows $\#\{\text{objective} \leq \tau * \text{best} \} / \# \text{instances}$, where objective corresponds to
the result of an algorithm on an instance and best refers to the best result of any algorithm.
The parameter $\tau$ in this equation is plotted on the $x$-axis.
For each algorithm, this yields a non-decreasing, piecewise constant function.
Thus, if we are interested in the number of instances where an algorithm is the best, we only need to look at $\tau = 1$.

\noindent \textbf{Benchmark Instances.} \label{Instances}
Experiments are performed using benchmark instances from the Polyhedral Benchmark suite (PolyBench)~\cite{PolyBench}
and the~ISPD98 VLSI Circuit Benchmark Suite~\cite{ISPD98}. The PolyBench instances were kindly provided to us by Herrmann \etal \cite{DBLP:journals/siamsc/HerrmannOUKC19} as DAGs. The ISPD98 graphs are based on the respective circuits and contain one node for each cell and a directed edge from the source of a net to each of its sinks. 
In case the resulting instance does not form a DAG, i.e., contains
cycles, we do the following: We gradually add directed edges and skip those that would create a cycle. 
Basic properties of the instances can be found in Appendix~\ref{app:benchmark_instances}. 
To perform experiments with DAHs, we transform all graphs into hypergraphs using the row-net model~\cite{PaToH}.
A hypergraph contains one hypernode for each node of the DAG and one hyperedge for each node $u$ that has outgoing edges. The head of the hyperedge is $u$ and the tails are the successors of $u$. \\

\begin{figure}[ht!]
\centering%
\includegraphics[page=1, height=4.5cm]{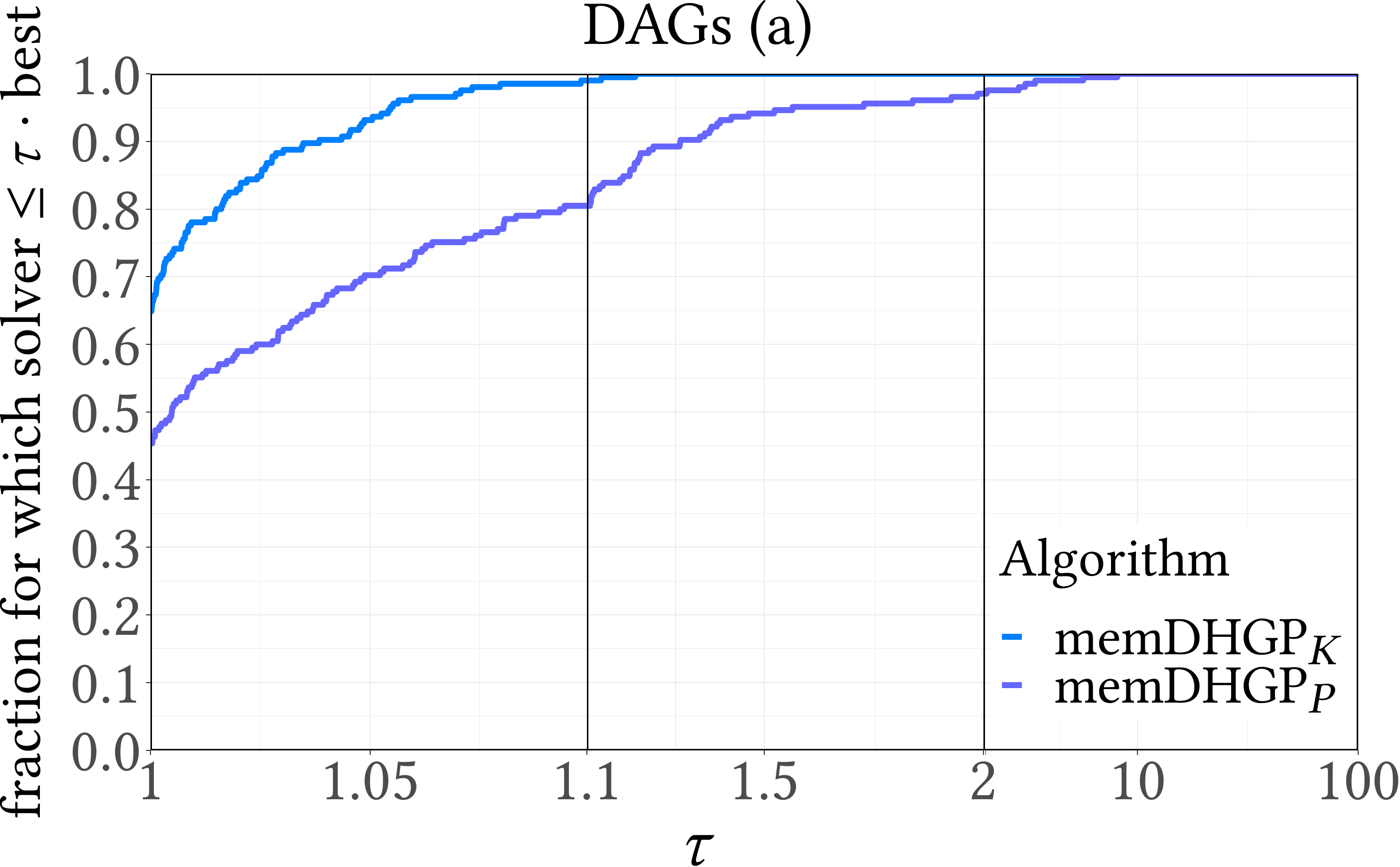}\\%
\vspace{0.3em}%
\includegraphics[page=1, height=4.5cm]{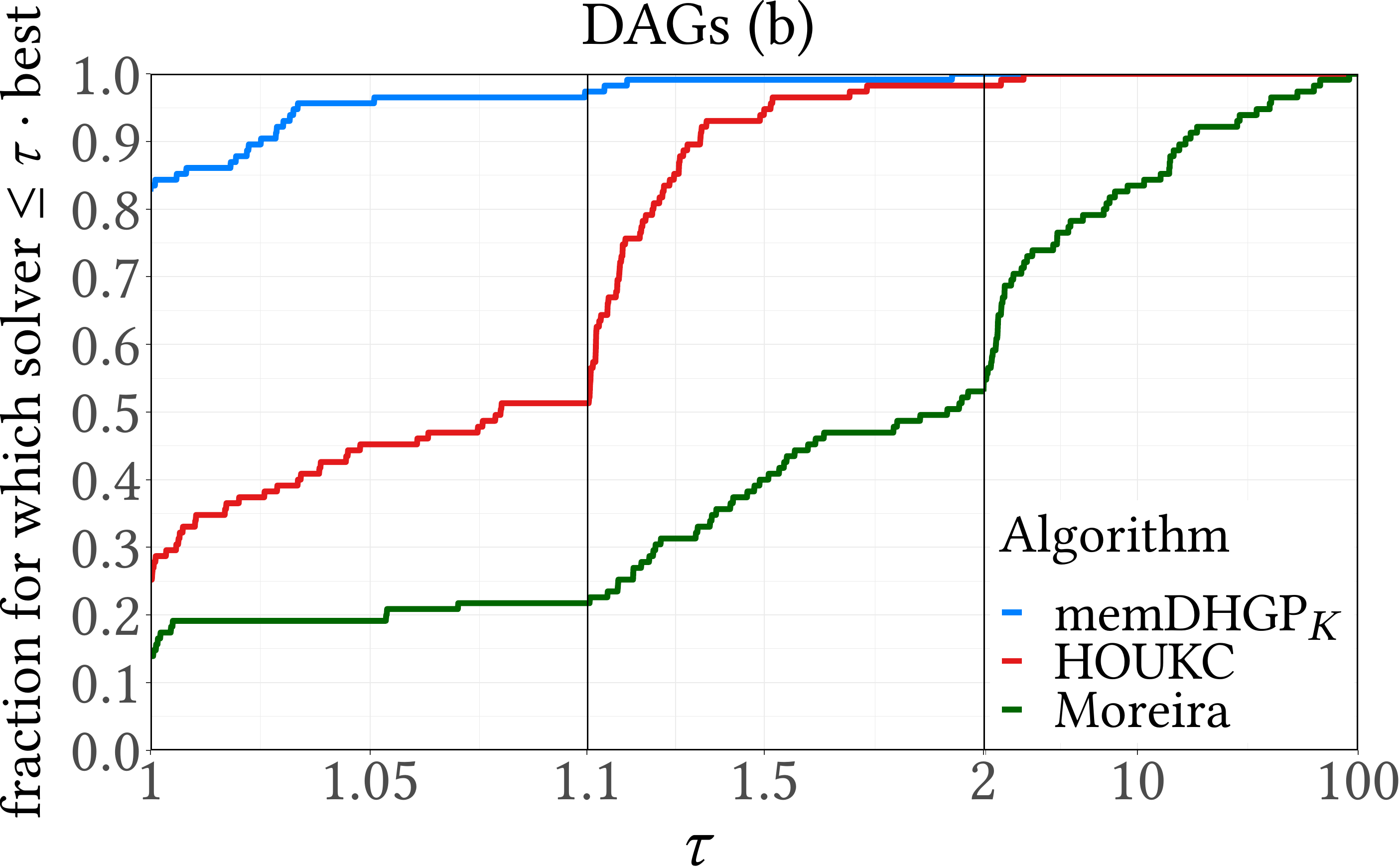}\\%
\vspace{0.3em}%
\includegraphics[page=1, height=4.5cm]{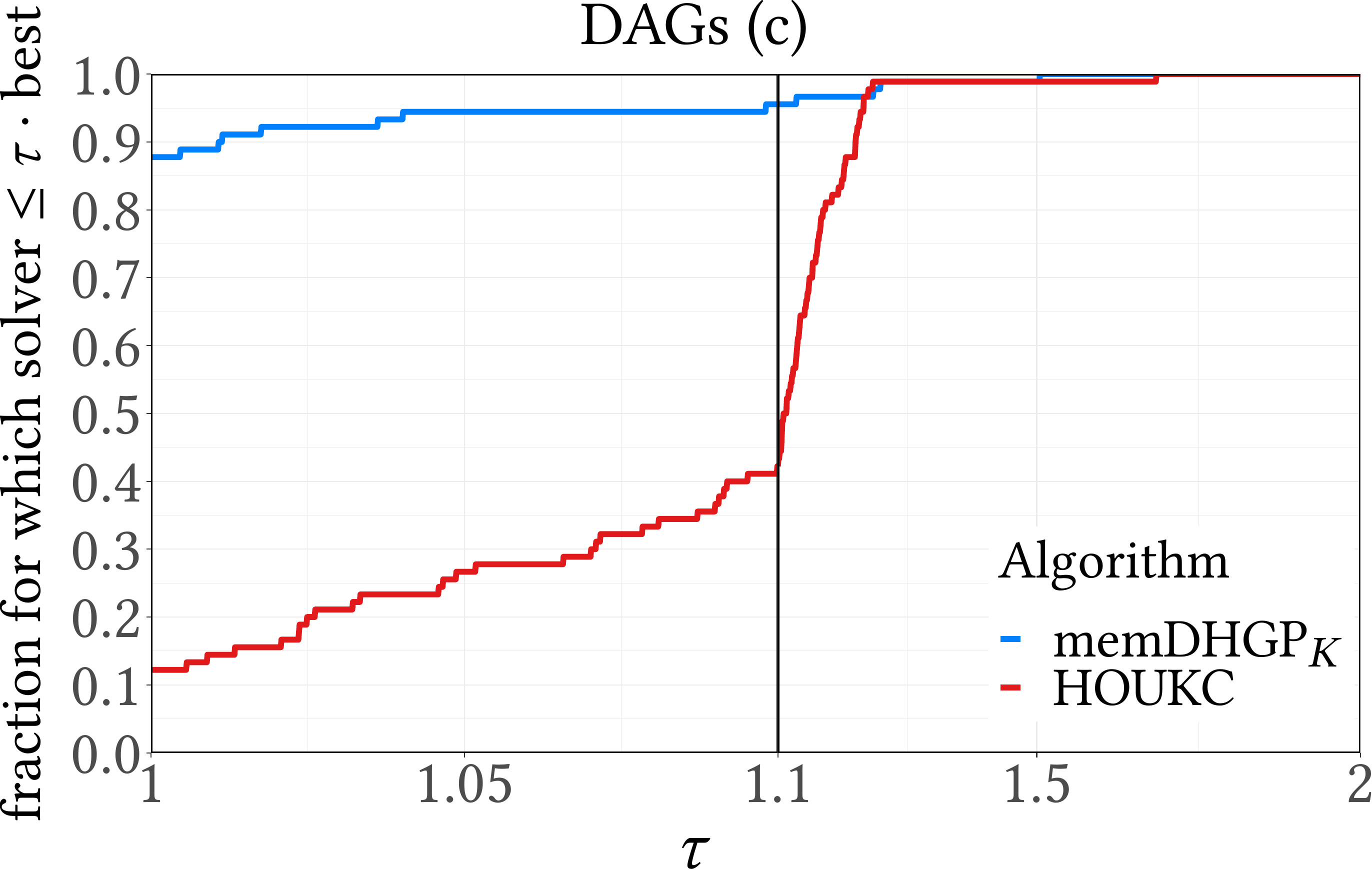}\\%
\vspace{0.3em}%
\includegraphics[page=1, height=4.5cm]{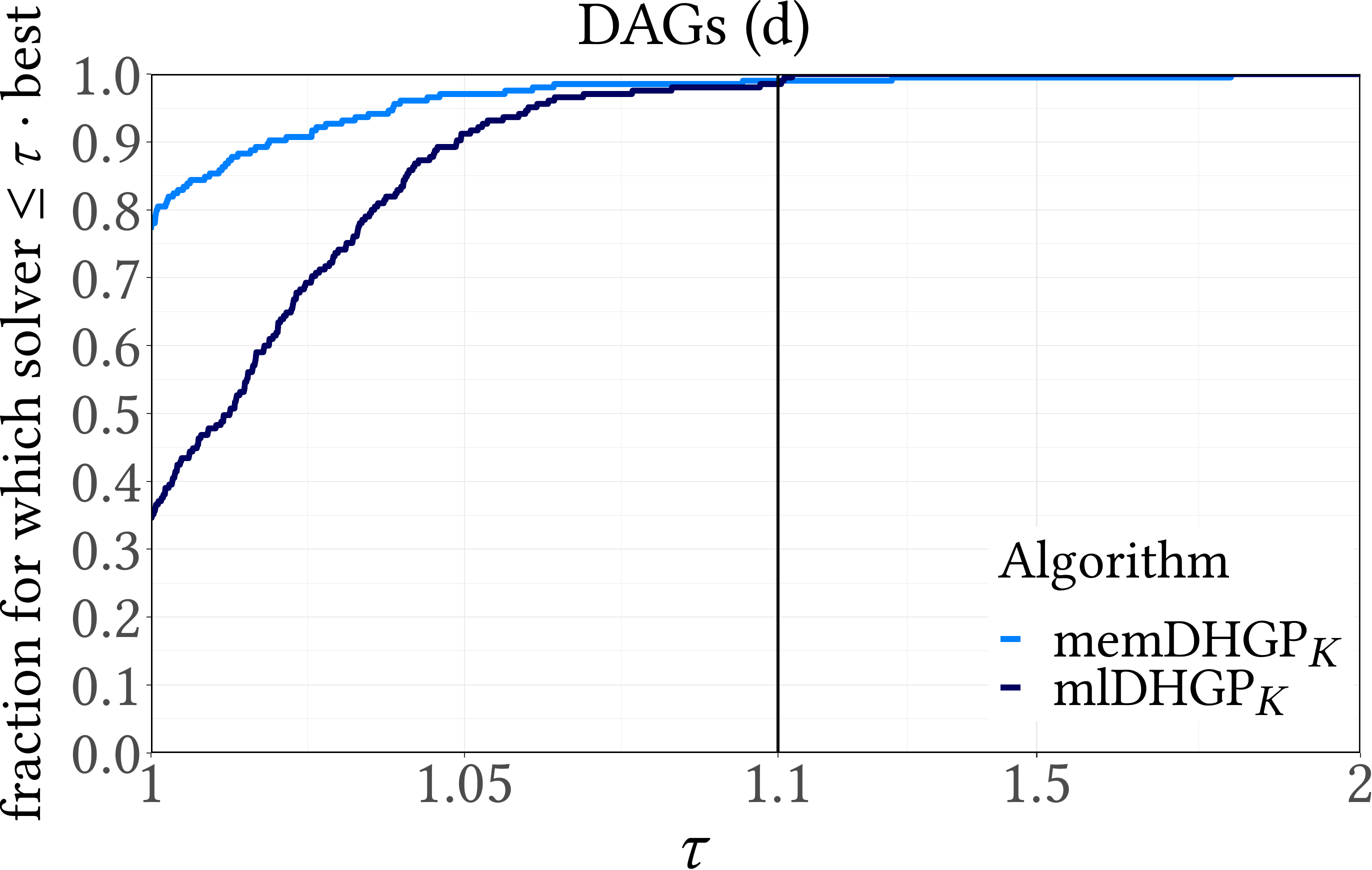}\\%
 
\caption{Performance profiles from top to bottom: (a) \texttt{memDHGP} with different initial partitioners on all instances as DAGs; (b) \texttt{HOUKC}, \texttt{memDHGP} and Moreira on PolyBench benchmark set as DAGs; (c) \texttt{HOUK}, and \texttt{memDHGP} memetic algorithms on ISPD98 benchmark set as DAGs; (d) Performance profiles on all instances as DAGs, for \texttt{memDHGP} and \texttt{mlDHGP} equipped with KaHyPar as initial partitioning algorithm.}
\vspace*{-.5cm}
\label{fig:expsummary}
\end{figure}

\begin{figure}[t]
\centering%
\includegraphics[height=4.5cm]{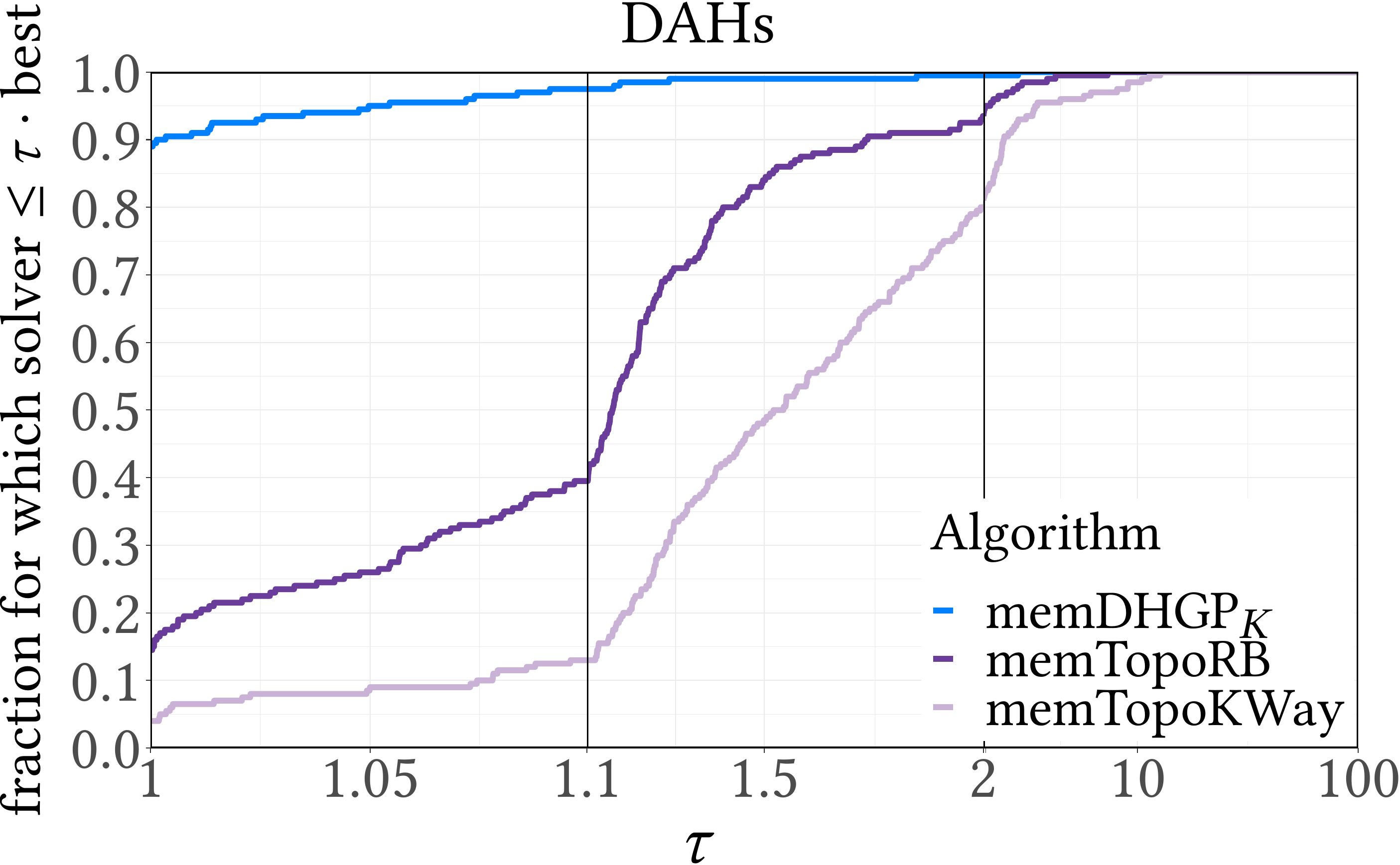}\\%
  
  \caption{Performance profile of \texttt{memDHGP}, \texttt{memTopoRB} and \texttt{memTopoKWay} on all instances as DAHs.}
  \vspace*{-.5cm}
  \label{fig:expsummary_dah}
\end{figure}

\subsection{Acyclic DAG Partitioning.} \label{subsec:eval}
We now focus on the DAG case, i.e., we evaluate all algorithms for the case that the input is a DAG (not a DAH).
\ifTR
Figure~\ref{fig:expsummary} summarizes the results while Tables~\ref{tab:detailedresults:cfg:1}-\ref{tab:detailedresults:ispd98:2} in Appendix~\ref{app:dags} give detailed per-instance results.
\fi
We focus on the case in which all solvers are given 8 hours of time to compute a solution (except for the adi instance, for which we give all solvers 24 hours) in order to make fair comparisons.
The detailed data also contains average results of single algorithm executions.
First, we compare different initial partitioning algorithms.
Overall, as can be seen in Figure~\ref{fig:expsummary} (a), KaHyPar seems to be the better choice as initial~partitioning~algorithm,
as it computes the best solutions for 65\% of all instances and is within 10\% of PaToHs solutions for almost all cases.

We now compare \texttt{memDHGP$_K$} with \texttt{HOUKC} and the evolutionary algorithm of Moreira~\etal~\cite{DBLP:conf/gecco/MoreiraP018}. Looking
at the results for the PolyBench instances in Figure~\ref{fig:expsummary} (b), we see that \texttt{memDHGP$_K$}
performs considerably better than both DAG partitioners -- computing the best solutions for more than $82\%$ of all
instances. Furthermore, it is within a factor of $1.1$ of the best algorithm in more than $95\%$ of all instances,
while \texttt{HOUKC} and Moreira are only within a factor of $1.1$ for $51\%$ (resp. $22\%$) of all instances. 
\ifTR
Considering the results presented in Tables~\ref{tab:detailedresults:cfg:1}-\ref{tab:detailedresults:ispd98:2},
we note that \texttt{memDHGP$_K$} computes 11.1\% better cuts on average than the previous state-of-the-art algorithm \texttt{HOUKC} on the PolyBench instances.
\fi 
It computes a better or equal result in \numprint{104} out of \numprint{115} cases and a strictly better result in \numprint{85} cases.
The largest improvement is observed on instance \texttt{covariance}, where Herrmann \etal compute
a result of \numprint{34 307} for $k=2$ and our algorithm computes a result of \numprint{11281}.
On the ISPD98 instances, our algorithm computes a better result in 79 out of 90 cases. On average, the result is improved by 9.7\% over \texttt{HOUKC}.

Figure~\ref{fig:expsummary} (d) shows the effect of the memetic algorithm. Here, we compare the results of \texttt{memDHGP$_K$} to the results
produced by our multilevel algorithm without the memetic component when run for 8 hours of time using different random seeds to compute a partition.
The experiments indicate that using a memetic strategy is more effective than repeated restarts of the algorithm.
Overall, \texttt{memDHGP$_K$} computes better cuts than \texttt{mlDHGP$_K$} for $76\%$ of all instances (i.e., in 155 out of 205 cases).
The largest improvement over repeated trials is 10\% on the PolyBench instances -- observed on instance 2mm for $k=8$ -- and 8\% on ibm01 for $k=2$.

\subsection{Acyclic DAH Partitioning.} \label{subsec:evalDAH}
We now switch to the DAH case. We exclude the adi instance from the PolyBench set since it is too large to be partitioned within a reasonable time using our multilevel algorithm.
Since \texttt{HOUKC} and Moreira are unable to partition DAHs and since there is \emph{no} previous work on the DAH partitioning problem that we are aware of, we compare \texttt{memDHGP} (using KaHyPar for initial partitioning) with two configurations that use different approaches to compute the initial population:
\texttt{memTopoRB} generates individuals by first using Kahn's algorithm to compute a topological order of the DAH, afterwards obtains two blocks by splitting along the ordering into two blocks, performs $2$-way local search, and then proceeds recursively. 
\texttt{memTopoKWay} generates individuals by computing one topological ordering, \emph{directly} splits the DAH into $k$ blocks, and then uses $k$-way local search as described above to refine the solution. 
For \texttt{mlDHGP}, \texttt{mlTopoRB}, and \texttt{mlTopoKWay}, the memetic algorithm is disabled, i.e.,
we just use the partitions computed by KaHyPar or topological ordering approaches as input to our multilevel algorithm. 
For \texttt{TopoRB} and \texttt{TopoKWay}, we furthermore disable the multilevel algorithm.
\ifTR
Table~\ref{tbl:dah_partitioning} summarizes the experiments, while Tables~\ref{app:dahs}.\ref{tab:detailedresultshg:cfg:1}-\ref{app:dahs}.\ref{tab:detailedresultshg:ispd98:2} give detailed per-instances results.
\fi 


\begin{table}[b!]
\vspace*{-.75cm}
	\centering
	  \caption{Geometric mean solution quality of the best results out of multiple repetitions for different algorithms on PolyBench and ISPD98 instances (as DAHs).}\label{tbl:dah_partitioning}
\vspace*{.25cm}
	\begin{tabular}{l|r}
	Algorithm & gmean \\
					  \midrule
	\texttt{TopoRB}  & 16\ 571 \\ 
	\texttt{mlTopoRB}& \numprint{9128} \\
	\texttt{memTopoRB} (8h) &8\ 071 \\
					  \midrule
										
	\texttt{TopoKWay}  &18 161 \\
	\texttt{mlTopoKWay}  & \numprint{13605} \\
	\texttt{memTopoKWay} (8h) &10\ 643 \\
					  \midrule
	\texttt{mlDHGP}  & \numprint{7244}  \\
	\texttt{memDHGP} (8h) &6\ 526  \\ 
	
	\end{tabular}
	\end{table}


Looking at Table~\ref{tbl:dah_partitioning} it can be seen that using recursive bisection has a significant advantage over using a direct $k$-way scheme. On average, \texttt{TopoRB} improves quality over \texttt{TopoKWay} by 9.6\%. 
This gets even more pronounced when using these algorithms as input to our multilevel algorithm where \texttt{mlTopoRB} improves partitions by 49\% over \texttt{mlTopoKWay}.
We believe that this is due to the fact that because of the acyclicity constraint, the $k$-way search space is much more fractured than the $2$-way search space.

Using KaHyPar as algorithm for initial partitioning in our multilevel algorithm (i.e., \texttt{mlDHGP}) improves the result significantly over \texttt{mlTopoRB}.
This is expected, since the multilevel algorithm adds a more global view to the optimization landscape.
The average improvement of \texttt{mlDHGP} is 26\% over \texttt{mlTopoRB} and 128.8\% over the single level \texttt{TopoRB} algorithm.
Switching to the memetic algorithm \texttt{memDHGP} improves the result over \texttt{mlDHGP} by another 11\% on average.
Using a high-quality algorithm to build the initial population has a clear advantage:
Considering Table~\ref{tbl:dah_partitioning}, \texttt{memTopoRB} and \texttt{memTopoKWay} compute 23.7\% and 63.1\% worse solutions on average. 

Looking at Figure~\ref{fig:expsummary_dah}, we see that \texttt{memDHGP} computes the best solutions for almost $90\%$ of all instances,
while the solutions of \texttt{memTopoRB} and \texttt{memTopoKWay} are more than a factor of 1.1 worse than the best
algorithm for more than 60\% (resp. 80\%) of all cases.

\begin{figure}[t!]
\centering%
\includegraphics[page=1, width=.75\columnwidth]{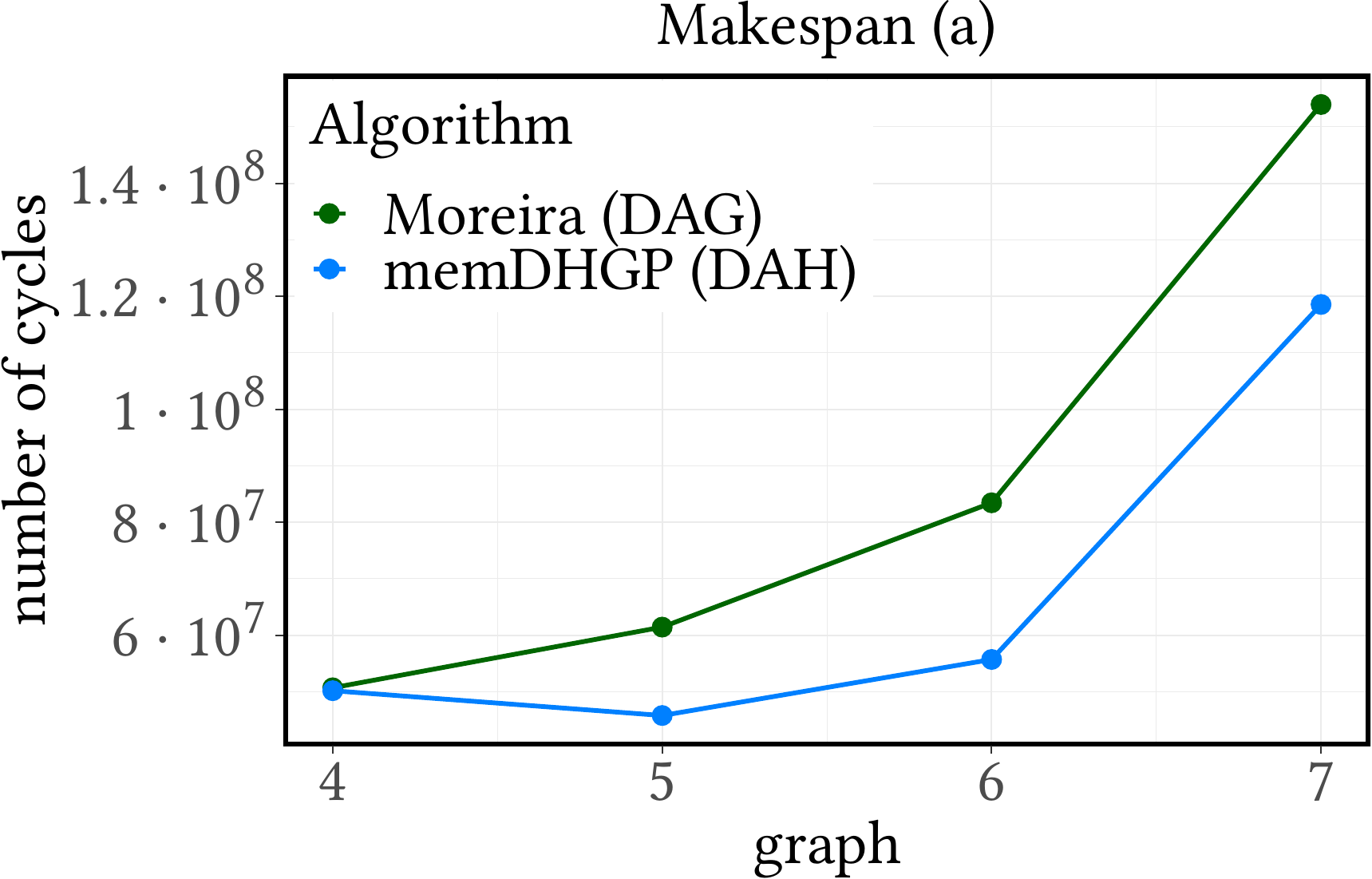}\\%
\vspace{0.3em}%
\includegraphics[page=1, width=.75\columnwidth]{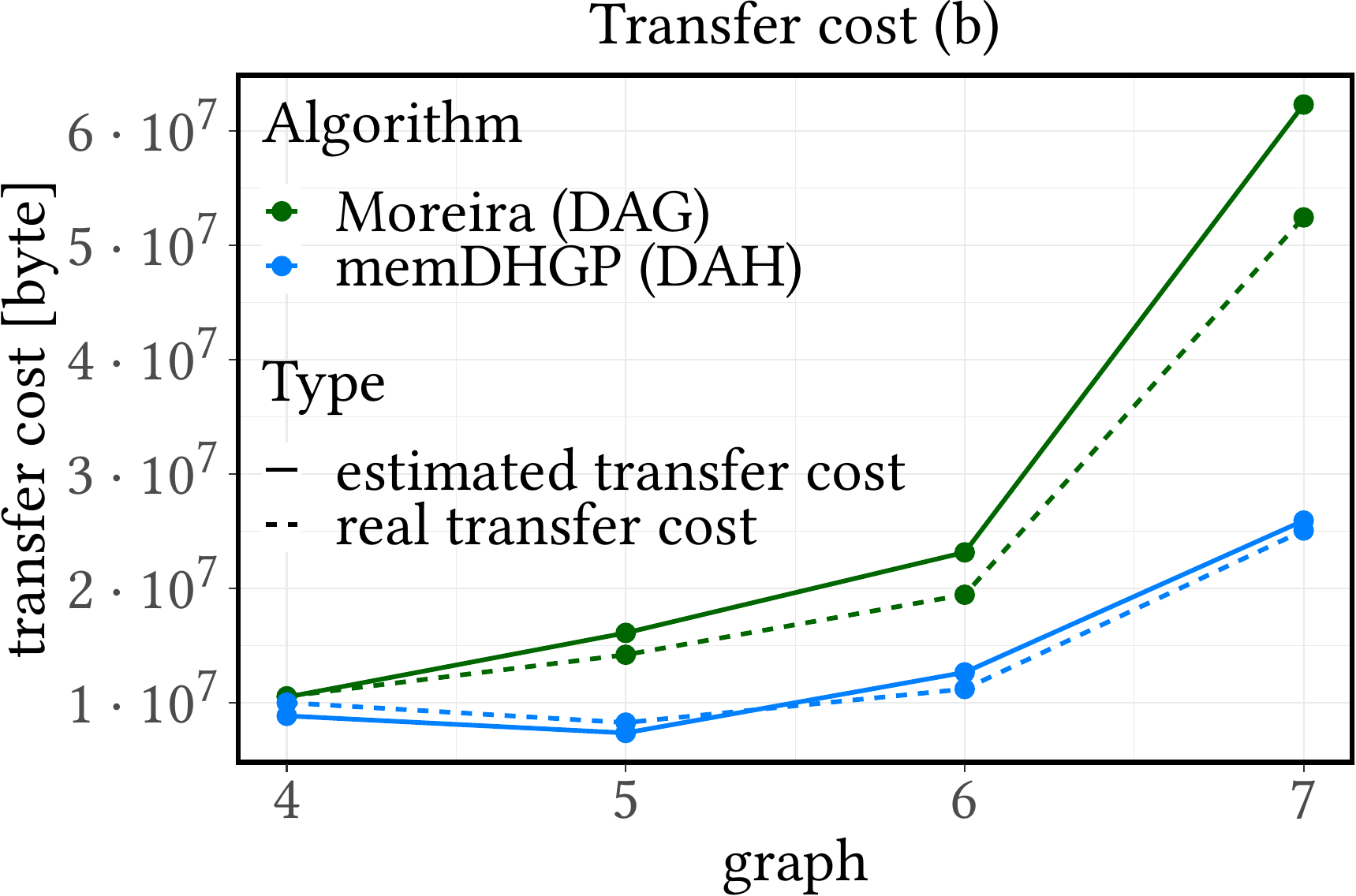}\\%
\caption{Makespan and transfer cost for partitions computed with \texttt{memDHGP} (DAH model) and the algorithm by Moreira \etal\cite{DBLP:conf/gecco/MoreiraP018} (DAG model).}
\label{fig:impactplots}
\vspace*{-.5cm}  
\end{figure}

Note that the overall improvement of \texttt{memDHGP} over the the simple \texttt{TopoRB} algorithm is 153\% on average.
We \emph{conclude} that (i) using recursive bisection is better than direct $k$-way methods,
(ii) that multilevel algorithms are indeed superior to single level algorithms,
(iii) that a high-quality initial solution (like the one obtained from using KaHyPar) is necessary to obtain high-quality overall in a multilevel (and memetic) algorithm,
and (iv) that our memetic strategy helps to effectively explore the search space.

\subsection{Impact on Imaging Application.} 
\label{sec:impact}
We evaluate the impact of DAH partitioning on
an advanced imaging algorithm, the \emph{Local Laplacian filter}~\cite{DBLP:journals/cacm/ParisHK15}.
We compare our \texttt{memDHGP} with the evolutionary DAG partitioner 
Moreira~\cite{DBLP:conf/gecco/MoreiraP018}.
We model the dataflow of the filter both as a DAG and a DAH.
In both cases, nodes are annotated with the program size and an execution time
estimate, and (hyper)edges with the corresponding data transfer size.
We use 4 different Laplacian filter configurations, which differ in the number of levels of image pyramids and remapping functions (parameter $K$).
A higher $K$ improves the resulting image quality but increases the (hyper)graph. We use $K \in \{4, 5, 6, 7\}$.
The DAG has $752$ nodes and $862$ hyperedges in total in its largest configuration for $K = 7$.
The time budget given to each heuristic is 10 minutes.
The makespans for the resulting schedules (Figure~\ref{fig:impactplots}~(a)) are obtained with a cycle-true compiled simulator of the hardware platform. We see that \texttt{memDHGP} outperforms Moreira.
The difference tends to increase for larger graphs, with the makespan for the largest filter being 22\% smaller.
Figure~\ref{fig:impactplots}~(b) compares the reported edge cuts to
the \emph{real} data transfer cost of the final application. It can be seen that in
the DAH case, the edge cut approximates the real transfer cost much better.
We can therefore conclude that hypergraphs are better suited to model data dependencies in this application domain.

\section{Conclusion and Future Work}

We engineered the first $n$-level algorithm for the acyclic hypergraph partitioning problem.
Based on this, we engineer a memetic algorithm to further reduce communication cost, as well as to improve scheduling makespan on embedded
multiprocessor architectures. Experiments indicate that our algorithm outperforms previous algorithms that focus on the directed acyclic graph case which have previously been employed in the application domain. 
Moreover, the results indicate that our algorithm is the current state-of-the-art for the DAH case.
Improvements stem from many places: first, using a multilevel scheme is better than using a single level algorithm; second,
using $n$ levels and a high-quality initial partitioning algorithm yields a significant advantage over previous algorithms that
employ standard multilevel algorithms; and lastly employing a memetic strategy further improves solution quality.
Lastly, we performed experiments that indicate that by using DAHs instead of DAGs yields the better model in the application domain.
Important future work includes parallelization of our algorithm.

\appendix
\bibliographystyle{abbrvnat}
\bibliography{p50-schlag}
\newpage

\begin{appendix}
\section{Benchmark Instances}\label{app:benchmark_instances}
\begin{table}[h!]
\label{tab:test_instances}
\begin{center}
\begin{tabular}{lrrr}

\toprule
PolyBench & $n$ & $m$ & Ref. \\
\midrule
2mm & \numprint{36500} & \numprint{62200} & \cite{PolyBench}\\
3mm & \numprint{111900} & \numprint{214600} & \cite{PolyBench}\\
adi & \numprint{596695} & \numprint{1059590} & \cite{PolyBench}\\
atax & \numprint{241730} & \numprint{385960} & \cite{PolyBench}\\
covariance & \numprint{191600} & \numprint{368775} & \cite{PolyBench}\\
doitgen & \numprint{123400} & \numprint{237000} & \cite{PolyBench}\\
durbin & \numprint{126246} & \numprint{250993} & \cite{PolyBench}\\
fdtd-2d & \numprint{256479} & \numprint{436580} & \cite{PolyBench}\\
gemm & \numprint{1026800} & \numprint{1684200} & \cite{PolyBench}\\
gemver & \numprint{159480} & \numprint{259440} & \cite{PolyBench}\\
gesummv & \numprint{376000} & \numprint{500500} & \cite{PolyBench}\\
heat-3d & \numprint{308480} & \numprint{491520} & \cite{PolyBench}\\
jacobi-1d & \numprint{239202} & \numprint{398000} & \cite{PolyBench}\\
jacobi-2d & \numprint{157808} & \numprint{282240} & \cite{PolyBench}\\
lu & \numprint{344520} & \numprint{676240} & \cite{PolyBench}\\
ludcmp & \numprint{357320} & \numprint{701680} & \cite{PolyBench}\\
mvt & \numprint{200800} & \numprint{320000} & \cite{PolyBench}\\
seidel-2d & \numprint{261520} & \numprint{490960} & \cite{PolyBench}\\
symm & \numprint{254020} & \numprint{440400} & \cite{PolyBench}\\
syr2k & \numprint{111000} & \numprint{180900} & \cite{PolyBench}\\
syrk & \numprint{594480} & \numprint{975240} & \cite{PolyBench}\\
trisolv & \numprint{240600} & \numprint{320000} & \cite{PolyBench}\\
trmm & \numprint{294570} & \numprint{571200} & \cite{PolyBench}\\
\midrule
ISPD98 & $n$ & $m$ & Ref. \\
\midrule
ibm01 & \numprint{13865} & \numprint{42767} & \cite{ISPD98}\\
ibm02 & \numprint{19325} & \numprint{61756} & \cite{ISPD98}\\
ibm03 & \numprint{27118} & \numprint{96152} & \cite{ISPD98}\\
ibm04 & \numprint{31683} & \numprint{108311} & \cite{ISPD98}\\
ibm05 & \numprint{27777} & \numprint{91478} & \cite{ISPD98}\\
ibm06 & \numprint{34660} & \numprint{97180} & \cite{ISPD98}\\
ibm07 & \numprint{47830} & \numprint{146513} & \cite{ISPD98}\\
ibm08 & \numprint{50227} & \numprint{265392} & \cite{ISPD98}\\
ibm09 & \numprint{60617} & \numprint{206291} & \cite{ISPD98}\\
ibm10 & \numprint{74452} & \numprint{299396} & \cite{ISPD98}\\
ibm11 & \numprint{81048} & \numprint{258875} & \cite{ISPD98}\\
ibm12 & \numprint{76603} & \numprint{392451} & \cite{ISPD98}\\
ibm13 & \numprint{99176} & \numprint{390710} & \cite{ISPD98}\\
ibm14 & \numprint{152255} & \numprint{480274} & \cite{ISPD98}\\
ibm15 & \numprint{186225} & \numprint{724485} & \cite{ISPD98}\\
ibm16 & \numprint{189544} & \numprint{648331} & \cite{ISPD98}\\
ibm17 & \numprint{188838} & \numprint{660960} & \cite{ISPD98}\\
ibm18 & \numprint{201648} & \numprint{597983} & \cite{ISPD98}\\
\bottomrule
\end{tabular}
\end{center}
\end{table}

\ifTR
\section{Detailed Experimental Results for DAG/DAH Partitioning}\label{app:dags}
\label{app:dahs}
Tables 2 to 5 refer to experiments with DAGs, Tables 6 to 9 refer to experiments with DAHs.
\begin{landscape}
\begin{table}[h!]
\scriptsize
\vspace*{-2cm}
\caption{Detailed per instance results on the PolyBench benchmark set \cite{PolyBench}. \emph{HOUKC} refers to the algorithm developed by Herrmann et.\,al.  \cite{DBLP:journals/siamsc/HerrmannOUKC19}. \emph{mlDHGP + X} refers to our algorithm with \emph{X} as undirected hypergraph partitioner for initial partitioning.  \emph{memDHGP + X} refers to our memetic algorithm that uses mlDHGP equiped with \emph{X} as undirected hypergraph partitioner for initial partitioning to build an initial population.  For HOUKC, the \emph{Average} column reports the better average from Table~A.1 and Table~A.2 in \cite{DBLP:journals/siamsc/HerrmannOUKC19} and the \emph{Best} column reports the best edge cut found during 8 hours of individual runs or the best edge cut reported in \cite{DBLP:journals/siamsc/HerrmannOUKC19}, if that is lower (marked with a star). For \emph{mlDHGP + X}, the \emph{Average} column reports the average edge cut of 5 individual runs and the \emph{Best} column reports the best edge cut found during 8 hours.  For \emph{memDHGP + X}, the \emph{Best} column reports the best result found after running for 8 hours.  The \emph{Overall Best} column shows the best cut found by any tool with the following identifiers: H: HOUKC, N: one of the new approaches, M: Moreira et.\,al. In general, lower is better.} \label{tab:detailedresults:cfg:1}

\begin{center}
\resizebox{.75\columnwidth}{!}{
\begin{tabular}{|lr|rr|rr|rrr|rr|rr|}
\hline 
& & \multicolumn{2}{c|}{HOUKC} & \multicolumn{2}{c|}{Moreira et.~al.} & \multicolumn{2}{c}{mlDHGP }& memDHGP & mlDHGP & memDHGP & \multicolumn{2}{c|}{Overall}\\ 
& &  & & & & \multicolumn{3}{c|}{with KaHyPar} & \multicolumn{2}{c|}{with PaToH} & \multicolumn{2}{c|}{Best}\\ 
\hline 
\multicolumn{1}{|c}{Graph} & \multicolumn{1}{c|}{K} & \multicolumn{1}{l}{Average} & \multicolumn{1}{c|}{Best (8h) or \cite{DBLP:journals/siamsc/HerrmannOUKC19}} & \multicolumn{1}{c}{Average} & \multicolumn{1}{c|}{Best} & \multicolumn{1}{c}{Average} & \multicolumn{1}{c}{Best (8h)} &\multicolumn{1}{c|}{Best (8h)} &\multicolumn{1}{c}{Average} &\multicolumn{1}{c|}{Best (8h)} & Result & Solver\\ 
\hline
\multirow{5}{*}{2mm}        & \numprint{2}  & \textbf{\numprint{200}}   & \textbf{\numprint{200}}   & \textbf{\numprint{200}}   & \textbf{\numprint{200}}   & \textbf{\numprint{200}}   & \textbf{\numprint{200}}   & \textbf{\numprint{200}}    & \textbf{\numprint{200}}   & \textbf{\numprint{200}}   & \numprint{200}      & H,M,N\\
                            & \numprint{4}  & \numprint{2160}           & \numprint{946}            & \numprint{947}            & \textbf{\numprint{930}}   & \numprint{1065}           & \textbf{\numprint{930}}   & \textbf{\numprint{930}}    & \numprint{1006}           & \textbf{\numprint{930}}   & \numprint{930}&M,N\\
                            & \numprint{8}  & \numprint{5361}           & \numprint{2910}           & \numprint{7181}           & \numprint{6604}           & \numprint{2819}           & \numprint{2576}           & \textbf{\numprint{2465}}   & \numprint{5563}           & \numprint{5110}           & \numprint{2465}&N\\
                            & \numprint{16} & \numprint{11196}          & \numprint{8103}           & \numprint{13330}          & \numprint{13092}          & \numprint{7090}           & \numprint{5963}           & \textbf{\numprint{5435}}   & \numprint{7881}           & \numprint{6632}           & \numprint{5435}&N\\
                            & \numprint{32} & \numprint{15911}          & \numprint{12708}          & \numprint{14583}          & \numprint{14321}          & \numprint{11397}          & \numprint{10635}          & \textbf{\numprint{10398}}  & \numprint{12228}          & \numprint{11012}          & \numprint{10398}&N\\
\hline                                                                                                                                                                                                                                           
\multirow{5}{*}{3mm}        & \numprint{2}  & \numprint{1000}           & \textbf{\numprint{800}}   & \numprint{1000}           & \numprint{1000}           & \textbf{\numprint{800}}   & \textbf{\numprint{800}}   & \textbf{\numprint{800}}    & \numprint{1000}           & \numprint{1000}           & \numprint{800}&H,N\\
                            & \numprint{4}  & \numprint{9264}           & \textbf{\numprint{2600}}  & \numprint{38722}          & \numprint{37899}          & \numprint{2647}           & \textbf{\numprint{2600}}  & \textbf{\numprint{2600}}   & \textbf{\numprint{2600}}  & \textbf{\numprint{2600}}  & \numprint{2600}&H,N\\
                            & \numprint{8}  & \numprint{24330}          & \numprint{7735}           & \numprint{58129}          & \numprint{49559}          & \numprint{8596}           & \numprint{6967}           & \textbf{\numprint{6861}}   & \numprint{14871}          & \numprint{9560}           & \numprint{6861}&N\\
                            & \numprint{16} & \numprint{37041}          & \numprint{21903}          & \numprint{64384}          & \numprint{60127}          & \numprint{23513}          & \textbf{\numprint{19625}} & \numprint{19675}           & \numprint{28021}          & \numprint{23967}          & \numprint{19675}&N\\
                            
                            & \numprint{32} & \numprint{46437}          & \numprint{36718}          & \numprint{62279}          & \numprint{58190}          & \numprint{34721}          & \textbf{\numprint{30908}} & \numprint{31423}           & \numprint{38879}          & \numprint{34353}          & \numprint{31423}&N\\
\hline              
\multirow{5}{*}{adi} & \numprint{2} & \numprint{142719} & *\textbf{\numprint{134675}} & \numprint{134945} & \textbf{\numprint{134675}} & \numprint{138433} & \numprint{138057} & \numprint{138279} &\numprint{138520} & \numprint{138329} & \numprint{134675} & H,M\\ 
& \numprint{4} & \numprint{212938} & \textbf{\numprint{210979}} & \numprint{284666} & \numprint{283892} & \numprint{213255} & \numprint{212709} & \numprint{212851} &\numprint{213390} & \numprint{212564} & \numprint{210979} & H\\ 
& \numprint{8} & \numprint{256302} & \textbf{\numprint{229563}} & \numprint{290823} & \numprint{290672} & \numprint{253885} & \numprint{252271} & \numprint{253206} &\numprint{254282} & \numprint{252376} & \numprint{229563} & H \\ 
& \numprint{16} & \numprint{282485} & \textbf{\numprint{271374}} & \numprint{326963} & \numprint{326923} & \numprint{281068} & \numprint{277337} & \numprint{280437} &\numprint{281751} & \numprint{276958} & \numprint{271374} & H\\ 
& \numprint{32} & \numprint{306075} & \numprint{305091} & \numprint{370876} & \numprint{370413} & \numprint{309930} & \numprint{303078} & \textbf{\numprint{299387}} &\numprint{309757} & \numprint{302157} & \numprint{302157} & N\\ 
\hline                                                                                                                                                                                                                             
\multirow{5}{*}{atax}       & \numprint{2}  & \numprint{39876}          & \numprint{32451}          & \numprint{47826}          & \numprint{47424}          & \numprint{39695}          & \numprint{24150}          & \textbf{\numprint{23690}}  & \numprint{45130}          & \numprint{43450}          & \numprint{23690}&N\\
                            & \numprint{4}  & \numprint{48645}          & \numprint{43511}          & \numprint{82397}          & \numprint{76245}          & \numprint{50725}          & \numprint{42028}          & \textbf{\numprint{39316}}  & \numprint{50144}          & \numprint{47486}          & \numprint{39316}&N\\
                            & \numprint{8}  & \numprint{51243}          & \numprint{48702}          & \numprint{113410}         & \numprint{111051}         & \numprint{54891}          & \numprint{48824}          & \textbf{\numprint{47741}}  & \numprint{52163}          & \numprint{49450}          & \numprint{47741}&N\\
                            & \numprint{16} & \numprint{59208}          & \numprint{52127}          & \numprint{127687}         & \numprint{125146}         & \numprint{68153}          & \textbf{\numprint{50962}} & \numprint{51256}           & \numprint{53256}          & \numprint{51191}          & \numprint{51191}&N\\
                            & \numprint{32} & \numprint{69556}          & \numprint{57930}          & \numprint{132092}         & \numprint{130854}         & \numprint{66267}          & \numprint{54613}          & \numprint{56051}           & \numprint{56773}          & \textbf{\numprint{54536}} & \numprint{54536}&N\\
\hline                                                                                                                                                                                                                                           
\multirow{5}{*}{covariance} & \numprint{2}  & \numprint{27269}          & \textbf{\numprint{4775}}  & \numprint{66520}          & \numprint{66445}          & \textbf{\numprint{4775}}  & \textbf{\numprint{4775}}  & \textbf{\numprint{4775}}   & \numprint{5893}           & \numprint{5641}           & \numprint{4775}&H,N\\
                            & \numprint{4}  & \numprint{61991}          & *\numprint{34307}         & \numprint{84626}          & \numprint{84213}          & \numprint{12362}          & \numprint{11724}          & \textbf{\numprint{11281}}  & \numprint{13339}          & \numprint{12344}          & \numprint{11281}&N\\
                            & \numprint{8}  & \numprint{74325}          & *\numprint{50680}         & \numprint{103710}         & \numprint{102425}         & \numprint{24429}          & \numprint{21460}          & \textbf{\numprint{21106}}  & \numprint{51984}          & \numprint{41807}          & \numprint{21106}&N\\
                            & \numprint{16} & \numprint{119284}         & \numprint{99629}          & \numprint{125816}         & \numprint{123276}         & \numprint{62011}          & \numprint{60143}          & \textbf{\numprint{58875}}  & \numprint{65302}          & \numprint{59153}          & \numprint{58875}&N\\
                            & \numprint{32} & \numprint{121155}         & \numprint{94247}          & \numprint{142214}         & \numprint{137905}         & \numprint{76977}          & \numprint{73758}          & \textbf{\numprint{72090}}  & \numprint{80464}          & \numprint{74770}          & \numprint{72090}&N\\
\hline                                                                                                                                                                                                                                           
\multirow{5}{*}{doitgen}    & \numprint{2}  & \numprint{5035}           & \textbf{\numprint{3000}}  & \numprint{43807}          & \numprint{42208}          & \textbf{\numprint{3000}}  & \textbf{\numprint{3000}}  & \textbf{\numprint{3000}}   & \textbf{\numprint{3000}}  & \textbf{\numprint{3000}}  & \numprint{3000}&H,N\\
                            & \numprint{4}  & \numprint{37051}          & \textbf{\numprint{9000}}  & \numprint{72115}          & \numprint{71082}          & \numprint{11029}          & \textbf{\numprint{9000}}  & \textbf{\numprint{9000}}   & \numprint{28317}          & \numprint{27852}          & \numprint{9000}&H,N\\
                            & \numprint{8}  & \numprint{51283}          & \numprint{36790}          & \numprint{76977}          & \numprint{75114}          & \numprint{36326}          & \numprint{34912}          & \textbf{\numprint{34682}}  & \numprint{42185}          & \numprint{38491}          & \numprint{34682}&N\\
                            & \numprint{16} & \numprint{62296}          & \numprint{50481}          & \numprint{84203}          & \numprint{77436}          & \numprint{51064}          & \numprint{48992}          & \numprint{50486}           & \numprint{50993}          & \textbf{\numprint{48193}} & \numprint{48193}&N\\
                            & \numprint{32} & \numprint{68350}          & \numprint{59632}          & \numprint{94135}          & \numprint{92739}          & \numprint{59159}          & \numprint{58184}          & \numprint{57408}           & \numprint{57208}          & \textbf{\numprint{55721}} & \numprint{55721}&N\\
\hline                                                                                                                                                                                                                                           
\multirow{5}{*}{durbin}     & \numprint{2}  & \textbf{\numprint{12997}} & \textbf{\numprint{12997}} & \textbf{\numprint{12997}} & \textbf{\numprint{12997}} & \textbf{\numprint{12997}} & \textbf{\numprint{12997}} & \textbf{\numprint{12997}}  & \textbf{\numprint{12997}} & \textbf{\numprint{12997}} & \numprint{12997}&H,M,N\\
                            & \numprint{4}  & \numprint{21566}          & *\numprint{21566}         & \numprint{21641}          & \numprint{21641}          & \numprint{21556}          & \numprint{21557}          & \textbf{\numprint{21541}}  & \numprint{21556}          & \textbf{\numprint{21541}} & \numprint{21541}&N\\
                            & \numprint{8}  & \numprint{27519}          & \numprint{27518}          & \numprint{27571}          & \numprint{27571}          & \numprint{27511}          & \textbf{\numprint{27508}} & \numprint{27509}           & \numprint{27511}          & \numprint{27509}          & \numprint{27509}&N\\
                            & \numprint{16} & \numprint{32852}          & \numprint{32841}          & \numprint{32865}          & \numprint{32865}          & \numprint{32869}          & \textbf{\numprint{32824}} & \numprint{32825}           & \numprint{32852}          & \numprint{32825}          & \numprint{32825}&N\\
                            & \numprint{32} & \numprint{39738}          & \numprint{39732}          & \numprint{39726}          & \numprint{39725}          & \numprint{39753}          & \numprint{39717}          & \textbf{\numprint{39701}}  & \numprint{39753}          & \textbf{\numprint{39701}} & \numprint{39701}&N\\
\hline                                                                                                                                                                                                                                           
\multirow{5}{*}{fdtd-2d}    & \numprint{2}  & \numprint{6024}           & \textbf{\numprint{4381}}  & \numprint{5494}           & \numprint{5494}           & \numprint{5233}           & \numprint{4756}           & \numprint{4604}            & \numprint{6318}           & \numprint{6285}           & \numprint{4381}&H\\
                            & \numprint{4}  & \numprint{15294}          & \numprint{11551}          & \numprint{15100}          & \numprint{15099}          & \numprint{11670}          & \numprint{9325}           & \textbf{\numprint{9240}}   & \numprint{11572}          & \numprint{10232}          & \numprint{9240}&N\\
                            & \numprint{8}  & \numprint{23699}          & \numprint{19527}          & \numprint{33087}          & \numprint{32355}          & \numprint{17704}          & \numprint{15906}          & \textbf{\numprint{15653}}  & \numprint{17990}          & \numprint{15758}          & \numprint{15653}&N\\
                            & \numprint{16} & \numprint{32917}          & \numprint{28065}          & \numprint{35714}          & \numprint{35239}          & \numprint{25170}          & \numprint{22866}          & \numprint{22041}           & \numprint{24582}          & \textbf{\numprint{22003}} & \numprint{22003}&N\\
                            & \numprint{32} & \numprint{42515}          & \numprint{39063}          & \numprint{43961}          & \numprint{42507}          & \numprint{32658}          & \numprint{30872}          & \numprint{29868}           & \numprint{32682}          & \textbf{\numprint{29772}} & \numprint{29772}&N\\
\hline                                                                                                                                                                                                                                           
\multirow{5}{*}{gemm}       & \numprint{2}  & \textbf{\numprint{4200}}  & \textbf{\numprint{4200}}  & \numprint{383084}         & \numprint{382433}         & \textbf{\numprint{4200}}  & \textbf{\numprint{4200}}  & \textbf{\numprint{4200}}   & \numprint{4768}           & \numprint{4690}           & \numprint{4200}&H,N\\
                            & \numprint{4}  & \numprint{59854}          & \textbf{\numprint{12600}} & \numprint{507250}         & \numprint{500526}         & \textbf{\numprint{12600}} & \textbf{\numprint{12600}} & \textbf{\numprint{12600}}  & \numprint{13300}          & \textbf{\numprint{12600}} & \numprint{12600}&H,N\\
                            & \numprint{8}  & \numprint{116990}         & \numprint{33382}          & \numprint{578951}         & \numprint{575004}         & \numprint{70827}          & \numprint{31413}          & \textbf{\numprint{30912}}  & \numprint{188172}         & \numprint{175495}         & \numprint{30912}&N\\
                            & \numprint{16} & \numprint{263050}         & \numprint{224173}         & \numprint{615342}         & \numprint{613373}         & \numprint{185872}         & \numprint{164235}         & \textbf{\numprint{148040}} & \numprint{202920}         & \numprint{194017}         & \numprint{148040}&N\\
                            & \numprint{32} & \numprint{330937}         & \numprint{277879}         & \numprint{626472}         & \numprint{623271}         & \numprint{270909}         & \numprint{265771}         & \textbf{\numprint{258607}} & \numprint{280849}         & \numprint{275188}         & \numprint{258607}&N\\
\hline                                                                                                                                                                                                                                           
\multirow{5}{*}{gemver}     & \numprint{2}  & \numprint{20913}          & *\numprint{20913}         & \numprint{29349}          & \numprint{29270}          & \numprint{22725}          & \numprint{19485}          & \numprint{19390}           & \numprint{20317}          & \textbf{\numprint{18930}} & \numprint{18930}&N\\
                            & \numprint{4}  & \numprint{40299}          & \numprint{35431}          & \numprint{49361}          & \numprint{49229}          & \numprint{38600}          & \numprint{35021}          & \textbf{\numprint{33324}}  & \numprint{37632}          & \numprint{34328}          & \numprint{33324}&N\\
                            & \numprint{8}  & \numprint{55266}          & \numprint{43716}          & \numprint{68163}          & \numprint{67094}          & \numprint{50440}          & \numprint{44253}          & \numprint{43276}           & \numprint{47799}          & \textbf{\numprint{42548}} & \numprint{42548}&N\\
                            & \numprint{16} & \numprint{59072}          & \numprint{54012}          & \numprint{78115}          & \numprint{75596}          & \numprint{53819}          & \numprint{48618}          & \numprint{48182}           & \numprint{53775}          & \textbf{\numprint{46563}} & \numprint{46563}&N\\
                            & \numprint{32} & \numprint{73131}          & \numprint{63012}          & \numprint{85331}          & \numprint{84865}          & \numprint{58898}          & \numprint{53581}          & \numprint{54953}           & \numprint{59210}          & \textbf{\numprint{52404}} & \numprint{52404}&N\\
\hline                                                                                                                                                                                                                                           
\multirow{5}{*}{gesummv}    & \numprint{2}  & \textbf{\numprint{500}}   & \textbf{\numprint{500}}   & \numprint{1666}           & \textbf{\numprint{500}}   & \textbf{\numprint{500}}   & \textbf{\numprint{500}}   & \textbf{\numprint{500}}    & \textbf{\numprint{500}}   & \textbf{\numprint{500}}   & \numprint{500}&H,M,N\\
                            & \numprint{4}  & \numprint{10316}          & \textbf{\numprint{1500}}  & \numprint{98542}          & \numprint{94493}          & \numprint{5096}           & \textbf{\numprint{1500}}  & \textbf{\numprint{1500}}   & \numprint{1548}           & \textbf{\numprint{1500}}  & \numprint{1500}&N\\
                            & \numprint{8}  & \numprint{9618}           & \numprint{4021}           & \numprint{101533}         & \numprint{98982}          & \numprint{25535}          & \textbf{\numprint{3500}}  & \textbf{\numprint{3500}}   & \numprint{3640}           & \textbf{\numprint{3500}}  & \numprint{3500}&N\\
                            & \numprint{16} & \numprint{35686}          & \numprint{11388}          & \numprint{112064}         & \numprint{104866}         & \numprint{30215}          & \textbf{\numprint{7500}}  & \textbf{\numprint{7500}}   & \numprint{7883}           & \textbf{\numprint{7500}}  & \numprint{7500}&N\\
                            & \numprint{32} & \numprint{45050}          & \numprint{28295}          & \numprint{117752}         & \numprint{114812}         & \numprint{31740}          & \numprint{15620}          & \numprint{16339}           & \numprint{16144}          & \textbf{\numprint{15500}} & \numprint{15500}&N\\
\hline                                                                                                                                                                                                                                           
\multirow{5}{*}{heat-3d}    & \numprint{2}  & \numprint{9378}           & \numprint{8936}           & \numprint{8695}           & \numprint{8684}           & \numprint{8930}           & \textbf{\numprint{8640}}  & \textbf{\numprint{8640}}   & \numprint{9242}           & \numprint{8936}           & \numprint{8640}&N\\
                            & \numprint{4}  & \numprint{16700}          & \numprint{15755}          & \textbf{\numprint{14592}} & \textbf{\numprint{14592}} & \numprint{15355}          & \numprint{14642}          & \textbf{\numprint{14592}}  & \numprint{16304}          & \numprint{14865}          & \numprint{14592}&M,N\\
                            & \numprint{8}  & \numprint{25883}          & \numprint{24326}          & \textbf{\numprint{20608}} & \textbf{\numprint{20608}} & \numprint{23307}          & \numprint{21190}          & \numprint{21300}           & \numprint{25462}          & \numprint{23074}          & \numprint{20608}&M\\
                            & \numprint{16} & \numprint{42137}          & *\numprint{41261}         & \numprint{31615}          & \textbf{\numprint{31500}} & \numprint{38909}          & \numprint{38053}          & \numprint{35909}           & \numprint{40148}          & \numprint{37659}          & \numprint{31500}&M\\
                            & \numprint{32} & \numprint{64614}          & \numprint{60215}          & \numprint{51963}          & \textbf{\numprint{50758}} & \numprint{55360}          & \numprint{53525}          & \numprint{51682}           & \numprint{54621}          & \numprint{50848}          & \numprint{50758}&M\\
\hline                                                                                                                                                                                                                                           
\end{tabular}}
\end{center}
\end{table}
\begin{table}[h!]
\scriptsize

\vspace*{-2cm}
\caption{Detailed per instance results on the PolyBench benchmark set \cite{PolyBench}. \emph{HOUKC} refers to the algorithm developed by Herrmann et.\,al.  \cite{DBLP:journals/siamsc/HerrmannOUKC19}. \emph{mlDHGP + X} refers to our algorithm with \emph{X} as undirected hypergraph partitioner for initial partitioning.  \emph{memDHGP + X} refers to our memetic algorithm that uses mlDHGP equiped with \emph{X} as undirected hypergraph partitioner for initial partitioning to build an initial population.  For HOUKC, the \emph{Average} column reports the better average from Table~A.1 and Table~A.2 in \cite{DBLP:journals/siamsc/HerrmannOUKC19} and the \emph{Best} column reports the best edge cut found during 8 hours of individual runs or the best edge cut reported in \cite{DBLP:journals/siamsc/HerrmannOUKC19}, if that is lower (marked with a star). For \emph{mlDHGP + X}, the \emph{Average} column reports the average edge cut of 5 individual runs and the \emph{Best} column reports the best edge cut found during 8 hours.  For \emph{memDHGP + X}, the \emph{Best} column reports the best result found after running for 8 hours.  The \emph{Overall Best} column shows the best cut found by any tool with the following identifiers: H: HOUKC, N: one of the new approaches, M: Moreira et.\,al. In general, lower is better.} 
\label{tab:detailedresults:cfg:2}

\begin{center}

\resizebox{.75\columnwidth}{!}{
\begin{tabular}{|lr|rr|rr|rrr|rr|rr|}
\hline 
& & \multicolumn{2}{c|}{HOUKC} & \multicolumn{2}{c|}{Moreira et.~al.} & \multicolumn{2}{c}{mlDHGP }& memDHGP & mlDHGP & memDHGP & \multicolumn{2}{c|}{Overall}\\ 
& &  & & & & \multicolumn{3}{c|}{with KaHyPar} & \multicolumn{2}{c|}{with PaToH} & \multicolumn{2}{c|}{Best}\\ 
\hline 
\multicolumn{1}{|c}{Graph} & \multicolumn{1}{c|}{K} & \multicolumn{1}{l}{Average} & \multicolumn{1}{c|}{Best (8h) or \cite{DBLP:journals/siamsc/HerrmannOUKC19}} & \multicolumn{1}{c}{Average} & \multicolumn{1}{c|}{Best} & \multicolumn{1}{c}{Average} & \multicolumn{1}{c}{Best (8h)} &\multicolumn{1}{c|}{Best (8h)} &\multicolumn{1}{c}{Average} &\multicolumn{1}{c|}{Best (8h)} & Result & Solver\\ 
\hline

\hline
\multirow{5}{*}{jacobi-1d}  & \numprint{2}  & \numprint{646}            & \textbf{\numprint{400}}   & \numprint{596}            & \numprint{596}            & \numprint{440}            & \textbf{\numprint{400}}   & \textbf{\numprint{400}}    & \numprint{491}            & \numprint{423}            & \numprint{400}&H,N\\
                            & \numprint{4}  & \numprint{1617}           & \numprint{1123}           & \numprint{1493}           & \numprint{1492}           & \numprint{1188}           & \numprint{1046}           & \textbf{\numprint{1044}}   & \numprint{1250}           & \numprint{1128}           & \numprint{1044}&N\\
                            & \numprint{8}  & \numprint{2845}           & \numprint{2052}           & \numprint{3136}           & \numprint{3136}           & \numprint{2028}           & \numprint{1754}           & \textbf{\numprint{1750}}   & \numprint{2170}           & \numprint{1855}           & \numprint{1750}&N\\
                            & \numprint{16} & \numprint{4519}           & \numprint{3517}           & \numprint{6340}           & \numprint{6338}           & \numprint{3140}           & \numprint{2912}           & \textbf{\numprint{2869}}   & \numprint{3355}           & \numprint{2982}           & \numprint{2869}&N\\
                            & \numprint{32} & \numprint{6742}           & \numprint{5545}           & \numprint{8923}           & \numprint{8750}           & \numprint{4776}           & \numprint{4565}           & \textbf{\numprint{4498}}   & \numprint{4910}           & \numprint{4587}           & \numprint{4498}&N\\
\hline

\multirow{5}{*}{jacobi-2d} & \numprint{2}  & \numprint{3445}   & *\numprint{3342}          & \numprint{2994}           & \numprint{2991}           & \numprint{3878}          & \numprint{3000}            & \textbf{\numprint{2986}}   & \numprint{3942}          & \numprint{3129}           & \numprint{2986}&N\\
                           & \numprint{4}  & \numprint{7370}   & \numprint{7243}           & \numprint{5701}           & \textbf{\numprint{5700}}  & \numprint{7591}          & \numprint{5979}            & \numprint{5881}            & \numprint{7528}          & \numprint{6245}           & \numprint{5700}&M\\
                           & \numprint{8}  & \numprint{13168}  & \numprint{12134}          & \numprint{9417}           & \numprint{9416}           & \numprint{10872}         & \numprint{9295}            & \textbf{\numprint{8935}}   & \numprint{11753}         & \numprint{10492}          & \numprint{8935}&N\\
                           & \numprint{16} & \numprint{21565}  & \numprint{18394}          & \numprint{16274}          & \numprint{16231}          & \numprint{15605}         & \numprint{14746}           & \textbf{\numprint{13867}}  & \numprint{15889}         & \numprint{14736}          & \numprint{13867}&N\\
                           & \numprint{32} & \numprint{29558}  & \numprint{25740}          & \numprint{22181}          & \numprint{21758}          & \numprint{20597}         & \numprint{19647}           & \textbf{\numprint{18979}}  & \numprint{21653}         & \numprint{19530}          & \numprint{18979}&N\\
\hline                                                                                                                                                                                                                                  
\multirow{5}{*}{lu}        & \numprint{2}  & \numprint{5351}   & \textbf{\numprint{4160}}  & \numprint{5210}           & \numprint{5162}           & \textbf{\numprint{4160}} & \textbf{\numprint{4160}}   & \textbf{\numprint{4160}}   & \textbf{\numprint{4160}} & \textbf{\numprint{4160}}  & \numprint{4160}&H,N\\
                           & \numprint{4}  & \numprint{21258}  & \textbf{\numprint{12214}} & \numprint{13528}          & \numprint{13510}          & \numprint{12720}         & \textbf{\numprint{12214}}  & \textbf{\numprint{12214}}  & \numprint{16091}         & \numprint{15992}          & \numprint{12214}&H,N\\
                           & \numprint{8}  & \numprint{53643}  & \numprint{34074}          & \numprint{33307}          & \textbf{\numprint{33211}} & \numprint{42963}         & \numprint{33873}           & \numprint{33954}           & \numprint{41113}         & \numprint{38318}          & \numprint{33211}&M\\
                           & \numprint{16} & \numprint{105289} & \numprint{81713}          & \numprint{74543}          & \textbf{\numprint{74006}} & \numprint{81224}         & \numprint{74400}           & \numprint{74448}           & \numprint{83980}         & \numprint{75150}          & \numprint{74006}&M\\
                           & \numprint{32} & \numprint{156187} & \numprint{141868}         & \numprint{130674}         & \numprint{129954}         & \numprint{125932}        & \numprint{122977}          & \textbf{\numprint{121451}} & \numprint{131850}        & \numprint{127904}         & \numprint{121451}&N\\
\hline                                                                                                                                                                                                                                  
\multirow{5}{*}{ludcmp}    & \numprint{2}  & \numprint{5731}   & \textbf{\numprint{5337}}  & \numprint{5380}           & \textbf{\numprint{5337}}  & \textbf{\numprint{5337}} & \textbf{\numprint{5337}}   & \textbf{\numprint{5337}}   & \textbf{\numprint{5337}} & \textbf{\numprint{5337}}  & \numprint{5337}&H,N\\
                           & \numprint{4}  & \numprint{22368}  & \numprint{15170}          & \textbf{\numprint{14744}} & \textbf{\numprint{14744}} & \numprint{18114}         & \numprint{16889}           & \numprint{17560}           & \numprint{26606}         & \numprint{17113}          & \numprint{14744}&N\\
                           & \numprint{8}  & \numprint{60255}  & \numprint{41086}          & \numprint{37228}          & \textbf{\numprint{37069}} & \numprint{46268}         & \numprint{37688}           & \numprint{37790}           & \numprint{52980}         & \numprint{39362}          & \numprint{37069}&N\\
                           & \numprint{16} & \numprint{106223} & \numprint{86959}          & \numprint{78646}          & \numprint{78467}          & \numprint{89958}         & \textbf{\numprint{76074}}  & \numprint{80706}           & \numprint{96275}         & \numprint{85572}          & \numprint{78467}&N\\
                           & \numprint{32} & \numprint{158619} & \numprint{144224}         & \numprint{134758}         & \numprint{134288}         & \numprint{130552}        & \textbf{\numprint{125957}} & \numprint{127454}          & \numprint{136218}        & \numprint{131161}         & \numprint{127454}&N\\
\hline                                                                                                                                                                                                                                  
\multirow{5}{*}{mvt}       & \numprint{2}  & \numprint{21281}  & \numprint{16768}          & \numprint{24528}          & \numprint{23091}          & \numprint{23798}         & \textbf{\numprint{16584}}  & \numprint{16596}           & \numprint{32856}         & \numprint{20016}          & \numprint{16596}&N\\
                           & \numprint{4}  & \numprint{38215}  & \textbf{\numprint{29229}} & \numprint{74386}          & \numprint{73035}          & \numprint{41156}         & \numprint{29318}           & \numprint{30070}           & \numprint{52353}         & \numprint{42870}          & \numprint{29229}&H,N\\
                           & \numprint{8}  & \numprint{46776}  & \numprint{39295}          & \numprint{86525}          & \numprint{82221}          & \numprint{50853}         & \numprint{36531}           & \textbf{\numprint{35471}}  & \numprint{60021}         & \numprint{55460}          & \numprint{35471}&N\\
                           & \numprint{16} & \numprint{54925}  & \numprint{48036}          & \numprint{99144}          & \numprint{97941}          & \numprint{58258}         & \textbf{\numprint{41727}}  & \numprint{42890}           & \numprint{65738}         & \numprint{59194}          & \numprint{42890}&N\\
                           & \numprint{32} & \numprint{62584}  & \numprint{54293}          & \numprint{105066}         & \numprint{104917}         & \numprint{58413}         & \textbf{\numprint{45958}}  & \numprint{46122}           & \numprint{69221}         & \numprint{64611}          & \numprint{46122}&N\\
\hline                                                                                                                                                                                                                                  
\multirow{5}{*}{seidel-2d} & \numprint{2}  & \numprint{4374}   & \textbf{\numprint{3401}}  & \numprint{4991}           & \numprint{4969}           & \numprint{4036}          & \numprint{3578}            & \numprint{3504}            & \numprint{4206}          & \numprint{3786}           & \numprint{3401}&H,N\\
                           & \numprint{4}  & \numprint{11784}  & \numprint{10872}          & \numprint{12197}          & \numprint{12169}          & \numprint{11352}         & \numprint{10645}           & \textbf{\numprint{10404}}  & \numprint{11480}         & \numprint{10604}          & \numprint{10404}&N\\
                           & \numprint{8}  & \numprint{21937}  & \numprint{20711}          & \numprint{21419}          & \numprint{21400}          & \numprint{19954}         & \numprint{18528}           & \textbf{\numprint{17770}}  & \numprint{20309}         & \numprint{18482}          & \numprint{17770}&N\\
                           & \numprint{16} & \numprint{38065}  & \numprint{33647}          & \numprint{38222}          & \numprint{38110}          & \numprint{29930}         & \numprint{27644}           & \textbf{\numprint{27583}}  & \numprint{30329}         & \numprint{28348}          & \numprint{27583}&N\\
                           & \numprint{32} & \numprint{58319}  & \numprint{51745}          & \numprint{52246}          & \numprint{51531}          & \numprint{41256}         & \numprint{38949}           & \textbf{\numprint{38175}}  & \numprint{42291}         & \numprint{39058}          & \numprint{38175}&N\\
\hline                                                                                                                                                                                                                                  
\multirow{5}{*}{symm}      & \numprint{2}  & \numprint{26374}  & \numprint{21963}          & \numprint{94357}          & \numprint{94214}          & \numprint{22000}         & \numprint{21840}           & \textbf{\numprint{21836}}  & \numprint{29871}         & \numprint{26134}          & \numprint{21836}&N\\
                           & \numprint{4}  & \numprint{59815}  & \numprint{42442}          & \numprint{127497}         & \numprint{126207}         & \numprint{41486}         & \numprint{38290}           & \textbf{\numprint{37854}}  & \numprint{65111}         & \numprint{57620}          & \numprint{37854}&N\\
                           & \numprint{8}  & \numprint{91892}  & \numprint{69554}          & \numprint{152984}         & \numprint{151168}         & \numprint{69569}         & \textbf{\numprint{58084}}  & \numprint{60644}           & \numprint{82865}         & \numprint{75151}          & \numprint{60644}&N\\
                           & \numprint{16} & \numprint{105418} & \numprint{89320}          & \numprint{167822}         & \numprint{167512}         & \numprint{90978}         & \textbf{\numprint{83703}}  & \numprint{85508}           & \numprint{96932}         & \numprint{89445}          & \numprint{85508}&N\\
                           & \numprint{32} & \numprint{108950} & \textbf{\numprint{97174}} & \numprint{174938}         & \numprint{174843}         & \numprint{110495}        & \numprint{104376}          & \numprint{100337}          & \numprint{108814}        & \numprint{104592}         & \numprint{97174}&H\\
\hline                                                                                                                                                                                                                                  
\multirow{5}{*}{syr2k}     & \numprint{2}  & \numprint{4343}   & \textbf{\numprint{900}}   & \numprint{11098}          & \numprint{3894}           & \textbf{\numprint{900}}  & \textbf{\numprint{900}}    & \textbf{\numprint{900}}    & \textbf{\numprint{900}}  & \textbf{\numprint{900}}   & \numprint{900}&H,N\\
                           & \numprint{4}  & \numprint{12192}  & \numprint{3048}           & \numprint{49662}          & \numprint{48021}          & \numprint{3150}          & \numprint{2978}            & \textbf{\numprint{2909}}   & \numprint{16589}         & \numprint{9991}           & \numprint{2909}&N\\
                           & \numprint{8}  & \numprint{28787}  & \numprint{12833}          & \numprint{57584}          & \numprint{57408}          & \numprint{12504}         & \textbf{\numprint{9969}}   & \numprint{10154}           & \numprint{21427}         & \numprint{19507}          & \numprint{10154}&N\\
                           & \numprint{16} & \numprint{29519}  & \numprint{24457}          & \numprint{59780}          & \numprint{59594}          & \numprint{25054}         & \textbf{\numprint{21626}}  & \numprint{21828}           & \numprint{26120}         & \numprint{23588}          & \numprint{21828}&N\\
                           & \numprint{32} & \numprint{36111}  & \numprint{31138}          & \numprint{60502}          & \numprint{60085}          & \numprint{33424}         & \numprint{31236}           & \numprint{29984}           & \numprint{31358}         & \textbf{\numprint{29340}} & \numprint{29340}&N\\
\hline                                                                                                                                                                                                                                  
\multirow{5}{*}{syrk}      & \numprint{2}  & \numprint{11740}  & \textbf{\numprint{3240}}  & \numprint{219263}         & \numprint{218019}         & \textbf{\numprint{3240}} & \textbf{\numprint{3240}}   & \textbf{\numprint{3240}}   & \numprint{3439}          & \textbf{\numprint{3240}}  & \numprint{3240}&H,N\\
                           & \numprint{4}  & \numprint{56832}  & \textbf{\numprint{9960}}  & \numprint{289509}         & \numprint{289088}         & \numprint{10417}         & \numprint{10119}           & \numprint{9970}            & \numprint{89457}         & \numprint{80801}          & \numprint{9960}&H\\
                           & \numprint{8}  & \numprint{112236} & \textbf{\numprint{30602}} & \numprint{329466}         & \numprint{327712}         & \numprint{83000}         & \numprint{46130}           & \numprint{58876}           & \numprint{107220}        & \numprint{101516}         & \numprint{30602}&H\\
                           & \numprint{16} & \numprint{179042} & \numprint{147058}         & \numprint{354223}         & \numprint{351824}         & \numprint{117357}        & \numprint{113122}          & \textbf{\numprint{111635}} & \numprint{150363}        & \numprint{135615}         & \numprint{111635}&N\\
                           & \numprint{32} & \numprint{196173} & \numprint{173550}         & \numprint{362016}         & \numprint{359544}         & \numprint{158590}        & \textbf{\numprint{154818}} & \numprint{154921}          & \numprint{182222}        & \numprint{175999}         & \numprint{154921}&N\\
\hline                                                                                                                                                                                                                                  
\multirow{5}{*}{trisolv}   & \numprint{2}  & \numprint{336}    & \numprint{280}            & \numprint{6788}           & \numprint{3549}           & \numprint{280}           & \textbf{\numprint{279}}    & \textbf{\numprint{279}}    & \numprint{308}           & \textbf{\numprint{279}}   & \numprint{279}&N\\
                           & \numprint{4}  & \numprint{828}    & \numprint{827}            & \numprint{43927}          & \numprint{43549}          & \numprint{823}           & \textbf{\numprint{821}}    & \textbf{\numprint{821}}    & \numprint{865}           & \numprint{823}            & \numprint{821}&N\\
                           & \numprint{8}  & \numprint{2156}   & \numprint{1907}           & \numprint{66148}          & \numprint{65662}          & \numprint{2112}          & \textbf{\numprint{1893}}   & \numprint{1895}            & \numprint{2035}          & \numprint{1897}           & \numprint{1895}&N\\
                           & \numprint{16} & \numprint{6240}   & \numprint{5285}           & \numprint{71838}          & \numprint{71447}          & \numprint{8719}          & \numprint{4125}            & \textbf{\numprint{4108}}   & \numprint{4358}          & \numprint{4240}           & \numprint{4108}&N\\
                           & \numprint{32} & \numprint{13431}  & *\numprint{13172}         & \numprint{79125}          & \numprint{79071}          & \numprint{16027}         & \numprint{8942}            & \numprint{8784}            & \numprint{9210}          & \textbf{\numprint{8716}}  & \numprint{8716}&N\\
\hline                                                                                                                                                                                                                                  
\multirow{5}{*}{trmm}      & \numprint{2}  & \numprint{13659}  & \textbf{\numprint{3440}}  & \numprint{138937}         & \numprint{138725}         & \textbf{\numprint{3440}} & \textbf{\numprint{3440}}   & \textbf{\numprint{3440}}   & \textbf{\numprint{3440}} & \textbf{\numprint{3440}}  &\numprint{3440}&H,N\\
                           & \numprint{4}  & \numprint{58477}  & \numprint{14543}          & \numprint{192752}         & \numprint{191492}         & \numprint{14942}         & \numprint{12622}           & \textbf{\numprint{12389}}  & \numprint{35964}         & \numprint{35824}          & \numprint{12389}&N\\
                           & \numprint{8}  & \numprint{92185}  & \numprint{49830}          & \numprint{225192}         & \numprint{223529}         & \numprint{65303}         & \numprint{46059}           & \textbf{\numprint{45053}}  & \numprint{67011}         & \numprint{61045}          & \numprint{45053}&N\\
                           & \numprint{16} & \numprint{128838} & \numprint{103975}         & \numprint{240788}         & \numprint{238159}         & \numprint{92172}         & \textbf{\numprint{79507}}  & \numprint{80186}           & \numprint{96421}         & \numprint{87275}          & \numprint{80186}&N\\
                           & \numprint{32} & \numprint{153644} & \numprint{131899}         & \numprint{246407}         & \numprint{245173}         & \numprint{120839}        & \numprint{115460}          & \textbf{\numprint{112267}} & \numprint{120753}        & \numprint{113205}         & \numprint{112267}&N\\
\hline
\hline
\textbf{Mean} & & 25777 & 17897 & \numprint{44923} & \numprint{43200} &\numprint{18887} & \numprint{15988} & \numprint{16095} & \numprint{20308} & \numprint{18642}& &\\ 
 
\hline
\end{tabular}}
\end{center}
\end{table}
\begin{table}
\scriptsize
\vspace*{-2cm}
\caption{Detailed per instance results on the ISPD98 benchmark set \cite{ISPD98}. \emph{HOUKC} refers to the algorithm developed by Herrmann et.\,al.  \cite{DBLP:journals/siamsc/HerrmannOUKC19}. \emph{mlDHGP + X} refers to our algorithm with \emph{X} as undirected hypergraph partitioner for initial partitioning.  \emph{memDHGP + X} refers to our memetic algorithm that uses mlDHGP equiped with \emph{X} as undirected hypergraph partitioner for initial partitioning to build an initial population.  The \emph{Best} column reports the best edge cut found during 8 hours of individual runs. For \emph{mlDHGP + X}, the \emph{Average} column reports the average edge cut of 5 individual runs and the \emph{Best} column reports the best edge cut found during 8 hours.  For \emph{memDHGP + X}, the \emph{Best} column reports the best result found after running for 8 hours.  The \emph{Overall Best} column shows the best cut found by any tool with the following identifiers: H: HOUKC, N: one of the new approaches. In general, lower is better.}
\label{tab:detailedresults:ispd98:1}

\begin{center}

\resizebox{.75\columnwidth}{!}{
\begin{tabular}{|lr|rr|rrr|rr|rr|}
\hline 
& & \multicolumn{2}{c|}{HOUKC} & \multicolumn{2}{c}{mlDHGP} & memDHGP & mlDHGP & memDHGP & \multicolumn{2}{c|}{Overall} \\ 
& & &  & \multicolumn{3}{c|}{with KaHyPar} & \multicolumn{2}{c|}{with PaToH} & \multicolumn{2}{c|}{Best} \\ 
\hline 
\multicolumn{1}{|c}{Graph} & \multicolumn{1}{c|}{K} & \multicolumn{1}{l}{Average} & \multicolumn{1}{l|}{Best (8h)} & \multicolumn{1}{l}{Average} &\multicolumn{1}{c}{Best (8h)} & \multicolumn{1}{r|}{Best (8h)} &\multicolumn{1}{l}{Average} &\multicolumn{1}{r|}{Best (8h)} & Result & Solver \\ 
\hline
\multirow{5}{*}{ibm01} & \numprint{2}  & \numprint{3175}   & \numprint{2752}           & \numprint{3235}  & \numprint{2428}           & \textbf{\numprint{2255}}  & \numprint{2730}  & \numprint{2290} & \numprint{2255} & N \\
                       & \numprint{4}  & \numprint{6092}   & \numprint{5099}           & \numprint{5434}  & \numprint{5028}           & \numprint{4848}           & \numprint{5325}  & \textbf{\numprint{4841}} & \numprint{4841} & N \\
                       & \numprint{8}  & \numprint{7449}   & \numprint{6880}           & \numprint{8026}  & \numprint{7240}           & \numprint{6958}           & \numprint{8268}  & \textbf{\numprint{6639}} & \numprint{6639} & N \\
                       & \numprint{16} & \numprint{10555}  & \numprint{8603}           & \numprint{9131}  & \numprint{8135}           & \numprint{8028}           & \numprint{8870}  & \textbf{\numprint{7627}} & \numprint{7627} & N \\
                       & \numprint{32} & \numprint{12652}  & \numprint{11119}          & \numprint{10909} & \numprint{10086}          & \numprint{9572}           & \numprint{11107} & \textbf{\numprint{9404}} & \numprint{9404} & N \\
\hline                                                                                                                                                            
\multirow{5}{*}{ibm02} & \numprint{2}  & \numprint{8540}   & \numprint{4708}           & \numprint{8772}  & \textbf{\numprint{3262}}  & \numprint{5873}           & \numprint{8806}  & \numprint{8599} & \numprint{3262} & N \\
                       & \numprint{4}  & \numprint{13264}  & \numprint{11375}          & \numprint{12290} & \textbf{\numprint{11374}} & \numprint{11497}          & \numprint{12317} & \numprint{11400} & \numprint{11374} & N \\
                       & \numprint{8}  & \numprint{17832}  & \numprint{16591}          & \numprint{17557} & \numprint{16522}          & \textbf{\numprint{16253}} & \numprint{17520} & \numprint{16387} & \numprint{16253} & N \\
                       & \numprint{16} & \numprint{24856}  & \numprint{23002}          & \numprint{21708} & \numprint{20209}          & \textbf{\numprint{19727}} & \numprint{22128} & \numprint{20455} & \numprint{19727} & N \\
                       & \numprint{32} & \numprint{30407}  & \numprint{29082}          & \numprint{26379} & \numprint{25263}          & \textbf{\numprint{24264}} & \numprint{26659} & \numprint{25393} & \numprint{24264} & N \\
\hline                                                                                                                                                            
\multirow{5}{*}{ibm03} & \numprint{2}  & \numprint{14601}  & \numprint{13687}          & \numprint{15278} & \numprint{12584}          & \textbf{\numprint{11870}} & \numprint{14265} & \numprint{12051} & \numprint{11870} & N \\
                       & \numprint{4}  & \numprint{21802}  & \numprint{20077}          & \numprint{20652} & \numprint{18622}          & \textbf{\numprint{17757}} & \numprint{18840} & \numprint{17835} & \numprint{17757} & N \\
                       & \numprint{8}  & \numprint{26051}  & \numprint{24361}          & \numprint{25370} & \numprint{21494}          & \textbf{\numprint{20579}} & \numprint{22975} & \numprint{20699} & \numprint{20579} & N \\
                       & \numprint{16} & \numprint{30776}  & \numprint{27238}          & \numprint{29885} & \numprint{24637}          & \numprint{24006}          & \numprint{28097} & \textbf{\numprint{23837}} & \numprint{23837} & N \\
                       & \numprint{32} & \numprint{33439}  & \numprint{31034}          & \numprint{32134} & \numprint{27309}          & \numprint{27093}          & \numprint{30035} & \textbf{\numprint{27085}} & \numprint{27085} & N \\
\hline                                                                                                                                                            
\multirow{5}{*}{ibm04} & \numprint{2}  & \numprint{9518}   & \numprint{9108}           & \numprint{9727}  & \numprint{8508}           & \textbf{\numprint{8237}}  & \numprint{9727}  & \textbf{\numprint{8237}} & \numprint{8237} & N \\
                       & \numprint{4}  & \numprint{14226}  & \numprint{13190}          & \numprint{12668} & \numprint{11512}          & \numprint{10970}          & \numprint{12358} & \textbf{\numprint{10944}} & \numprint{10944} & N \\
                       & \numprint{8}  & \numprint{18508}  & \numprint{16683}          & \numprint{18677} & \numprint{16983}          & \numprint{16298}          & \numprint{18811} & \textbf{\numprint{15878}} & \numprint{15878} & N \\
                       & \numprint{16} & \numprint{25885}  & \numprint{22874}          & \numprint{24363} & \numprint{22800}          & \numprint{21812}          & \numprint{24298} & \textbf{\numprint{21373}} & \numprint{21373} & N \\
                       & \numprint{32} & \numprint{30512}  & \numprint{27107}          & \numprint{27882} & \numprint{26486}          & \textbf{\numprint{25078}} & \numprint{28127} & \numprint{25680} & \numprint{25078} & N \\
\hline                                                                                                                                                            
\multirow{5}{*}{ibm05} & \numprint{2}  & \numprint{8360}   & \numprint{5882}           & \numprint{7494}  & \textbf{\numprint{5478}}  & \numprint{5830}           & \numprint{7285}  & \numprint{6979} & \numprint{5478} & N \\
                       & \numprint{4}  & \numprint{17040}  & \numprint{13278}          & \numprint{14932} & \numprint{10740}          & \textbf{\numprint{10710}} & \numprint{15035} & \numprint{11885} & \numprint{10710} & N \\
                       & \numprint{8}  & \numprint{23170}  & \numprint{19480}          & \numprint{19618} & \numprint{16076}          & \numprint{15980}          & \numprint{19803} & \textbf{\numprint{15934}} & \numprint{15934} & N \\
                       & \numprint{16} & \numprint{29747}  & \numprint{25590}          & \numprint{25512} & \numprint{22049}          & \textbf{\numprint{20771}} & \numprint{24914} & \numprint{21604} & \numprint{20771} & N \\
                       & \numprint{32} & \numprint{34495}  & \numprint{30721}          & \numprint{29437} & \numprint{27465}          & \numprint{27582}          & \numprint{30155} & \textbf{\numprint{26899}} & \numprint{26899} & N \\
\hline                                                                                                                                                            
\multirow{5}{*}{ibm06} & \numprint{2}  & \numprint{14049}  & \numprint{12736}          & \numprint{12664} & \numprint{11804}          & \numprint{11341}          & \numprint{12832} & \textbf{\numprint{11285}} & \numprint{11285} & N \\
                       & \numprint{4}  & \numprint{23206}  & \numprint{20317}          & \numprint{21641} & \numprint{19097}          & \textbf{\numprint{18197}} & \numprint{21705} & \numprint{18374} & \numprint{18197} & N \\
                       & \numprint{8}  & \numprint{30875}  & \numprint{26980}          & \numprint{25402} & \numprint{23202}          & \numprint{22455}          & \numprint{25155} & \textbf{\numprint{22263}} & \numprint{22263} & N \\
                       & \numprint{16} & \numprint{34069}  & \numprint{30848}          & \numprint{29421} & \numprint{27435}          & \textbf{\numprint{26384}} & \numprint{29793} & \numprint{27263} & \numprint{26384} & N \\
                       & \numprint{32} & \numprint{38243}  & \numprint{36197}          & \numprint{32781} & \numprint{31310}          & \numprint{30839}          & \numprint{32826} & \textbf{\numprint{30597}} & \numprint{30597} & N \\
\hline                                                                                                                                                            
\multirow{5}{*}{ibm07} & \numprint{2}  & \numprint{16341}  & \numprint{15855}          & \numprint{15738} & \numprint{15356}          & \numprint{13681}          & \numprint{16003} & \textbf{\numprint{12965}} & \numprint{12965} & N \\
                       & \numprint{4}  & \numprint{26842}  & \numprint{23522}          & \numprint{22608} & \numprint{21583}          & \numprint{20499}          & \numprint{22273} & \textbf{\numprint{20348}} & \numprint{20348} & N \\
                       & \numprint{8}  & \numprint{29702}  & \numprint{27069}          & \numprint{26935} & \numprint{25655}          & \textbf{\numprint{24464}} & \numprint{27186} & \numprint{24586} & \numprint{24464} & N \\
                       & \numprint{16} & \numprint{36633}  & \numprint{33606}          & \numprint{31746} & \numprint{30788}          & \numprint{29808}          & \numprint{32195} & \textbf{\numprint{29797}} & \numprint{29797} & N \\
                       & \numprint{32} & \numprint{43083}  & \numprint{40205}          & \numprint{36959} & \numprint{35901}          & \textbf{\numprint{34648}} & \numprint{37017} & \numprint{34665} & \numprint{34648} & N \\
\hline                                                                                                                                                            
\multirow{5}{*}{ibm08} & \numprint{2}  & \numprint{25139}  & \numprint{24481}          & \numprint{24418} & \numprint{22381}          & \numprint{22079}          & \numprint{24384} & \textbf{\numprint{21925}} & \numprint{21925} & N \\
                       & \numprint{4}  & \numprint{52118}  & \numprint{38711}          & \numprint{41350} & \numprint{38644}          & \numprint{38495}          & \numprint{41402} & \textbf{\numprint{38330}} & \numprint{38330} & N \\
                       & \numprint{8}  & \numprint{84639}  & \numprint{81587}          & \numprint{50063} & \numprint{49238}          & \numprint{48429}          & \numprint{50043} & \textbf{\numprint{47124}} & \numprint{47124} & N \\
                       & \numprint{16} & \numprint{96107}  & \numprint{88135}          & \numprint{88727} & \numprint{87323}          & \textbf{\numprint{85996}} & \numprint{89513} & \numprint{86083} & \numprint{85996} & N \\
                       & \numprint{32} & \numprint{109264} & \numprint{96746}          & \numprint{93556} & \numprint{92591}          & \numprint{90779}          & \numprint{94172} & \textbf{\numprint{90660}} & \numprint{90660} & N \\
\hline                                                                                                                                                            
\multirow{5}{*}{ibm09} & \numprint{2}  & \numprint{19509}  & \numprint{15084}          & \numprint{17233} & \numprint{12661}          & \numprint{12305}          & \numprint{16307} & \textbf{\numprint{12127}} & \numprint{12127} & N \\
                       & \numprint{4}  & \numprint{28408}  & \numprint{25120}          & \numprint{26143} & \numprint{23342}          & \numprint{22557}          & \numprint{26184} & \textbf{\numprint{20892}} & \numprint{20892} & N \\
                       & \numprint{8}  & \numprint{36168}  & \numprint{31734}          & \numprint{33276} & \numprint{30411}          & \textbf{\numprint{29654}} & \numprint{34341} & \numprint{30168} & \numprint{29654} & N \\
                       & \numprint{16} & \numprint{42373}  & \numprint{39154}          & \numprint{39712} & \numprint{37301}          & \numprint{35902}          & \numprint{39529} & \textbf{\numprint{34707}} & \numprint{34707} & N \\
                       & \numprint{32} & \numprint{50041}  & \numprint{45987}          & \numprint{45226} & \numprint{41007}          & \numprint{40701}          & \numprint{45131} & \textbf{\numprint{39887}} & \numprint{39887} & N \\
\hline
\end{tabular}}
\end{center}
\end{table}
\begin{table}[h!]
\scriptsize
\vspace*{-2cm}
\caption{Detailed per instance results on the ISPD98 benchmark set \cite{ISPD98}. \emph{HOUKC} refers to the algorithm developed by Herrmann et.\,al.  \cite{DBLP:journals/siamsc/HerrmannOUKC19}. \emph{mlDHGP + X} refers to our algorithm with \emph{X} as undirected hypergraph partitioner for initial partitioning.  \emph{memDHGP + X} refers to our memetic algorithm that uses mlDHGP equiped with \emph{X} as undirected hypergraph partitioner for initial partitioning to build an initial population.  The \emph{Best} column reports the best edge cut found during 8 hours of individual runs. For \emph{mlDHGP + X}, the \emph{Average} column reports the average edge cut of 5 individual runs and the \emph{Best} column reports the best edge cut found during 8 hours.  For \emph{memDHGP + X}, the \emph{Best} column reports the best result found after running for 8 hours.  The \emph{Overall Best} column shows the best cut found by any tool with the following identifiers: H: HOUKC, N: one of the new approaches. In general, lower is better.}
\label{tab:detailedresults:ispd98:2}

\begin{center}
\resizebox{.75\columnwidth}{!}{
\begin{tabular}{|lr|rr|rrr|rr|rr|}
\hline 
& & \multicolumn{2}{c|}{HOUKC} & \multicolumn{2}{c}{mlDHGP} & memDHGP & mlDHGP & memDHGP & \multicolumn{2}{c|}{Overall}  \\ 
& & &  & \multicolumn{3}{c|}{with KaHyPar} & \multicolumn{2}{c|}{with PaToH} & \multicolumn{2}{c|}{Best} \\ 
\hline 
\multicolumn{1}{|c}{Graph} & \multicolumn{1}{c|}{K} & \multicolumn{1}{l}{Average} & \multicolumn{1}{l|}{Best (8h)} & \multicolumn{1}{l}{Average} &\multicolumn{1}{c}{Best (8h)} & \multicolumn{1}{r|}{Best (8h)} &\multicolumn{1}{l}{Average} &\multicolumn{1}{r|}{Best (8h)} & Result & Solver \\ 
\hline                                                                                                                                                            
\multirow{5}{*}{ibm10} & \numprint{2}  & \numprint{24983}  & \numprint{24073}          & \numprint{24310} & \numprint{21575}          & \numprint{21328}          & \numprint{22560} & \textbf{\numprint{21310}} & \numprint{21310} & N \\
                       & \numprint{4}  & \numprint{38620}  & \numprint{35083}          & \numprint{39383} & \numprint{33217}          & \numprint{36352}          & \numprint{39288} & \textbf{\numprint{32101}} & \numprint{32101} & N \\
                       & \numprint{8}  & \numprint{49646}  & \numprint{44820}          & \numprint{47827} & \numprint{40423}          & \numprint{39202}          & \numprint{46082} & \textbf{\numprint{38238}} & \numprint{38238} & N \\
                       & \numprint{16} & \numprint{63960}  & \numprint{54164}          & \numprint{55610} & \numprint{50854}          & \textbf{\numprint{49632}} & \numprint{56129} & \numprint{49892} & \numprint{49632} & N \\
                       & \numprint{32} & \numprint{69990}  & \numprint{65302}          & \numprint{64229} & \numprint{61838}          & \numprint{59914}          & \numprint{64105} & \textbf{\numprint{59180}} & \numprint{59180} & N \\
\hline                                                                                                                                                            
\multirow{5}{*}{ibm11} & \numprint{2}  & \numprint{19224}  & \numprint{16926}          & \numprint{21879} & \numprint{14374}          & \numprint{13578}          & \numprint{15748} & \textbf{\numprint{13318}} & \numprint{13318} & N \\
                       & \numprint{4}  & \numprint{36346}  & \numprint{26539}          & \numprint{26919} & \numprint{22750}          & \numprint{21623}          & \numprint{24724} & \textbf{\numprint{21310}} & \numprint{21310} & N \\
                       & \numprint{8}  & \numprint{39755}  & \numprint{32812}          & \numprint{32816} & \numprint{30401}          & \numprint{28563}          & \numprint{33247} & \textbf{\numprint{28477}} & \numprint{28477} & N \\
                       & \numprint{16} & \numprint{52698}  & \numprint{45779}          & \numprint{40706} & \numprint{38055}          & \numprint{39294}          & \numprint{43773} & \textbf{\numprint{37257}} & \numprint{37257} & N \\
                       & \numprint{32} & \numprint{63925}  & \numprint{57699}          & \numprint{50612} & \numprint{47999}          & \textbf{\numprint{47331}} & \numprint{52963} & \numprint{47930} & \numprint{47331} & N \\
\hline                                                                                                                                                            
\multirow{5}{*}{ibm12} & \numprint{2}  & \numprint{29359}  & \textbf{\numprint{27238}} & \numprint{30315} & \numprint{27860}          & \numprint{27365}          & \numprint{29620} & \numprint{27688} & \numprint{27238} & H \\
                       & \numprint{4}  & \numprint{50457}  & \numprint{47922}          & \numprint{49225} & \numprint{44108}          & \textbf{\numprint{42728}} & \numprint{49591} & \numprint{46107} & \numprint{42728} & N \\
                       & \numprint{8}  & \numprint{60024}  & \numprint{53785}          & \numprint{57394} & \numprint{52487}          & \numprint{51425}          & \numprint{57046} & \textbf{\numprint{50955}} & \numprint{50955} & N \\
                       & \numprint{16} & \numprint{72429}  & \numprint{65979}          & \numprint{66486} & \numprint{62965}          & \textbf{\numprint{61186}} & \numprint{67160} & \numprint{61484} & \numprint{61186} & N \\
                       & \numprint{32} & \numprint{84328}  & \numprint{76066}          & \numprint{73872} & \numprint{70503}          & \numprint{68739}          & \numprint{73252} & \textbf{\numprint{68712}} & \numprint{68712} & N \\

\hline

\multirow{5}{*}{ibm13} & \numprint{2}  & \numprint{30698}  & \numprint{19008}          & \numprint{21700}  & \textbf{\numprint{17161}}  & \numprint{17484}           & \numprint{22151}  & \numprint{17659} & \numprint{17161} & N \\
                       & \numprint{4}  & \numprint{39781}  & \numprint{29198}          & \numprint{39288}  & \numprint{31700}           & \numprint{32060}           & \numprint{38609}  & \textbf{\numprint{26500}} & \numprint{26500} & N \\
                       & \numprint{8}  & \numprint{54061}  & \textbf{\numprint{39453}} & \numprint{55253}  & \numprint{42881}           & \numprint{44535}           & \numprint{47765}  & \numprint{41596} & \numprint{39453} & H \\
                       & \numprint{16} & \numprint{71208}  & \numprint{60006}          & \numprint{65263}  & \numprint{55070}           & \textbf{\numprint{49820}}  & \numprint{65962}  & \numprint{49993} & \numprint{49820} & N \\
                       & \numprint{32} & \numprint{89053}  & \numprint{76762}          & \numprint{81831}  & \numprint{72262}           & \numprint{74997}           & \numprint{81416}  & \textbf{\numprint{70987}} & \numprint{70987} & N \\
\hline                                                                                                                                                               
\multirow{5}{*}{ibm14} & \numprint{2}  & \numprint{33205}  & \textbf{\numprint{31988}} & \numprint{51511}  & \numprint{48338}           & \numprint{48140}           & \numprint{52065}  & \numprint{49670} & \numprint{31988} & H \\
                       & \numprint{4}  & \numprint{55342}  & \textbf{\numprint{49972}} & \numprint{69320}  & \numprint{64838}           & \numprint{62888}           & \numprint{70364}  & \numprint{66680} & \numprint{49972} & H \\
                       & \numprint{8}  & \numprint{76297}  & \numprint{68992}          & \numprint{68051}  & \numprint{62718}           & \numprint{60929}           & \numprint{67598}  & \textbf{\numprint{56972}} & \numprint{56972} & N \\
                       & \numprint{16} & \numprint{96638}  & \numprint{80591}          & \numprint{79801}  & \numprint{74705}           & \textbf{\numprint{73224}}  & \numprint{80029}  & \numprint{73861} & \numprint{73224} & N \\
                       & \numprint{32} & \numprint{104543} & \numprint{96677}          & \numprint{91692}  & \numprint{89688}           & \numprint{87904}           & \numprint{92823}  & \textbf{\numprint{86504}} & \numprint{86504} & N \\
\hline                                                                                                                                                               
\multirow{5}{*}{ibm15} & \numprint{2}  & \numprint{74713}  & \numprint{71593}          & \numprint{66301}  & \numprint{63603}           & \textbf{\numprint{63136}}  & \numprint{82679}  & \numprint{67804} & \numprint{63136} & N \\
                       & \numprint{4}  & \numprint{105577} & \numprint{95911}          & \numprint{97786}  & \textbf{\numprint{87849}}  & \numprint{92812}           & \numprint{96479}  & \numprint{88349} & \numprint{87849} & N \\
                       & \numprint{8}  & \numprint{146984} & \numprint{123993}         & \numprint{123403} & \numprint{112014}          & \numprint{113564}          & \numprint{124884} & \textbf{\numprint{108619}} & \numprint{108619} & N \\
                       & \numprint{16} & \numprint{169587} & \numprint{153693}         & \numprint{136151} & \numprint{135061}          & \textbf{\numprint{124709}} & \numprint{143941} & \numprint{133614} & \numprint{124709} & N \\
                       & \numprint{32} & \numprint{191476} & \numprint{174057}         & \numprint{158765} & \numprint{154660}          & \numprint{149558}          & \numprint{160815} & \textbf{\numprint{148763}} & \numprint{148763} & N \\
\hline                                                                                                                                                               
\multirow{5}{*}{ibm16} & \numprint{2}  & \numprint{55871}  & \numprint{52980}          & \numprint{51699}  & \numprint{48222}           & \numprint{48063}           & \numprint{50167}  & \textbf{\numprint{45371}} & \numprint{45371} & N \\
                       & \numprint{4}  & \numprint{108576} & \numprint{93874}          & \numprint{98471}  & \numprint{93941}           & \numprint{91481}           & \numprint{99729}  & \textbf{\numprint{89976}} & \numprint{89976} & N \\
                       & \numprint{8}  & \numprint{130302} & \numprint{117375}         & \numprint{129900} & \numprint{115437}          & \numprint{119439}          & \numprint{126431} & \textbf{\numprint{114458}} & \numprint{114458} & N \\
                       & \numprint{16} & \numprint{162743} & \numprint{148626}         & \numprint{147987} & \numprint{136916}          & \numprint{134387}          & \numprint{142235} & \textbf{\numprint{132412}} & \numprint{132412} & N \\
                       & \numprint{32} & \numprint{181924} & \numprint{172909}         & \numprint{166347} & \numprint{158854}          & \numprint{157879}          & \numprint{164966} & \textbf{\numprint{153490}} & \numprint{153490} & N \\
\hline                                                                                                                                                               
\multirow{5}{*}{ibm17} & \numprint{2}  & \numprint{75860}  & \numprint{57177}          & \numprint{70331}  & \numprint{59100}           & \numprint{59470}           & \numprint{61401}  & \textbf{\numprint{56895}} & \numprint{56895} & N \\
                       & \numprint{4}  & \numprint{100287} & \numprint{89849}          & \numprint{121023} & \numprint{78692}           & \textbf{\numprint{77889}}  & \numprint{121175} & \numprint{107211} & \numprint{77889} & N \\
                       & \numprint{8}  & \numprint{151126} & \numprint{141679}         & \numprint{152455} & \textbf{\numprint{124639}} & \numprint{126610}          & \numprint{147848} & \numprint{130307} & \numprint{124639} & N \\
                       & \numprint{16} & \numprint{182272} & \numprint{166847}         & \numprint{171507} & \numprint{153812}          & \numprint{155789}          & \numprint{165498} & \textbf{\numprint{150026}} & \numprint{150026} & N \\
                       & \numprint{32} & \numprint{211541} & \numprint{198404}         & \numprint{188792} & \textbf{\numprint{167274}} & \numprint{173762}          & \numprint{194056} & \numprint{182853} & \numprint{167274} & N \\
\hline                                                                                                                                                               
\multirow{5}{*}{ibm18} & \numprint{2}  & \numprint{37123}  & \numprint{34949}          & \numprint{35907}  & \numprint{33434}           & \numprint{33394}           & \numprint{36651}  & \textbf{\numprint{33277}} & \numprint{33277} & N \\
                       & \numprint{4}  & \numprint{63000}  & \numprint{53948}          & \numprint{64540}  & \numprint{53190}           & \numprint{53237}           & \numprint{58432}  & \textbf{\numprint{48482}} & \numprint{48482} & N \\
                       & \numprint{8}  & \numprint{92636}  & \numprint{78164}          & \numprint{86580}  & \numprint{76686}           & \numprint{75728}           & \numprint{81435}  & \textbf{\numprint{70558}} & \numprint{70558} & N \\
                       & \numprint{16} & \numprint{121219} & \numprint{108744}         & \numprint{107824} & \numprint{93018}           & \textbf{\numprint{88959}}  & \numprint{113181} & \numprint{98976} & \numprint{88959} & N \\
                       & \numprint{32} & \numprint{144219} & \numprint{132289}         & \numprint{124788} & \numprint{111650}          & \textbf{\numprint{110816}} & \numprint{128875} & \numprint{119170} & \numprint{110816} & N \\
\hline                                                                                                                                                               
\hline                                                                                                                                                               
\textbf{Mean}          &               & 41189             & 36205                     & \numprint{37828}  & \numprint{33459}           & \numprint{33007}           & \numprint{37382}  & \numprint{33088} & & \\
\hline
\end{tabular}}
\end{center}
\end{table}
\label{app:dahs}
\begin{table}[h!]
\begin{center}
\scriptsize
        \vspace*{-2cm}
\caption{Detailed per instance results on the PolyBench benchmark set \cite{PolyBench}. \emph{mlDHGP} refers to our algorithm with \emph{KaHyPar} as undirected hypergraph partitioner for initial partitioning.  \emph{memDHGP} refers to our memetic algorithm that uses mlDHGP equipped with \emph{KaHyPar} as undirected hypergraph partitioner for initial partitioning to build an initial population.  The \emph{Best} column reports the best edge cut found during 8 hours of individual runs. Instances that took longer than 8 hours to compute are marked with a star. For \emph{mlDHGP}, the \emph{Average} column reports the average edge cut of 5 individual runs.  For \emph{memDHGP}, the \emph{Best} column reports the best result found after running for 8 hours.  In general, lower is better.}
\label{tab:detailedresultshg:cfg:1}

\resizebox{.5\columnwidth}{!}{ 
        \begin{tabular}{|lr|rr|rr|rr|}
\hline 
& & mlDHGP & memDHGP & \multicolumn{2}{c}{TopoRB} & \multicolumn{2}{|c|}{TopoKWay} \\
\hline 
\multicolumn{1}{|c}{Hypergraph} & \multicolumn{1}{c|}{K} & \multicolumn{1}{r}{Average} & \multicolumn{1}{r|}{Best (8h)} & \multicolumn{1}{r}{Average} & \multicolumn{1}{r|}{Best (8h)} &\multicolumn{1}{r}{Average} & \multicolumn{1}{r|}{Best (8h)}\\
\hline 
\multirow{5}{*}{2mm} & \numprint{2} & \numprint{212} & \textbf{\numprint{200}} & \numprint{224} & \textbf{\numprint{200}} & \numprint{344} & \numprint{210} \\
& \numprint{4} & \numprint{633} & \textbf{\numprint{608}} & \numprint{905} & \numprint{840} & \numprint{1618} & \numprint{750} \\
& \numprint{8} & \numprint{1376} & \textbf{\numprint{1320}} & \numprint{1608} & \numprint{1440} & \numprint{3169} & \numprint{1433} \\
& \numprint{16} & \numprint{2239} & \textbf{\numprint{2153}} & \numprint{2695} & \numprint{2248} & \numprint{4691} & \numprint{2630} \\
& \numprint{32} & \numprint{3796} & \textbf{\numprint{3624}} & \numprint{4562} & \numprint{3934} & \numprint{7015} & \numprint{4229} \\
\hline 
\multirow{5}{*}{3mm} & \numprint{2} & \textbf{\numprint{800}} & \textbf{\numprint{800}} & \numprint{1112} & \numprint{805} & \numprint{1564} & \numprint{1090} \\
& \numprint{4} & \numprint{2419} & \textbf{\numprint{2000}} & \numprint{3155} & \numprint{2480} & \numprint{5036} & \numprint{3566} \\
& \numprint{8} & \numprint{3950} & \textbf{\numprint{3540}} & \numprint{5940} & \numprint{4689} & \numprint{9374} & \numprint{5653} \\
& \numprint{16} & \numprint{6264} & \textbf{\numprint{5999}} & \numprint{9099} & \numprint{7537} & \numprint{12996} & \numprint{8123} \\
& \numprint{32} & \numprint{9234} & \textbf{\numprint{8861}} & \numprint{12719} & \numprint{11483} & \numprint{19224} & \numprint{12516} \\
\hline 
\multirow{5}{*}{atax} & \numprint{2} & \numprint{9206} & \textbf{\numprint{460}} & \textbf{\numprint{460}} & \textbf{\numprint{460}} & \numprint{14644} & \numprint{5829} \\
& \numprint{4} & \numprint{9438} & \numprint{4943} & \numprint{7162} & \textbf{\numprint{1719}} & \numprint{24248} & \numprint{19462} \\
& \numprint{8} & \numprint{22036} & \numprint{17127} & \numprint{20110} & \textbf{\numprint{9291}} & \numprint{27736} & \numprint{20983} \\
& \numprint{16} & \numprint{30917} & \numprint{28378} & \numprint{29675} & \textbf{\numprint{24167}} & \numprint{46152} & \numprint{29036} \\
& \numprint{32} & \numprint{43936} & \numprint{41981} & \numprint{40637} & \textbf{\numprint{39098}} & \numprint{52265} & \numprint{46790} \\
\hline 
\multirow{5}{*}{covariance} & \numprint{2} & \numprint{2930} & \textbf{\numprint{2590}} & \numprint{3343} & \numprint{3160} & \numprint{3190} & \numprint{3059} \\
& \numprint{4} & \numprint{6058} & \numprint{5705} & \numprint{5361} & \textbf{\numprint{5265}} & \numprint{7029} & \numprint{5681} \\
& \numprint{8} & \numprint{8834} & \textbf{\numprint{8238}} & \numprint{9660} & \numprint{9092} & \numprint{12815} & \numprint{10472} \\
& \numprint{16} & \numprint{13406} & \textbf{\numprint{12758}} & \numprint{13917} & \numprint{13480} & \numprint{19825} & \numprint{16529} \\
& \numprint{32} & \numprint{17605} & \textbf{\numprint{17210}} & \numprint{20211} & \numprint{19833} & \numprint{29596} & \numprint{24640} \\
\hline 
\multirow{5}{*}{doitgen} & \numprint{2} & \textbf{\numprint{400}} & \textbf{\numprint{400}} & \numprint{3134} & \numprint{2927} & \numprint{3444} & \numprint{2283} \\
& \numprint{4} & \textbf{\numprint{1200}} & \textbf{\numprint{1200}} & \numprint{3652} & \numprint{3600} & \numprint{6760} & \numprint{3114} \\
& \numprint{8} & \numprint{2892} & \textbf{\numprint{2800}} & \numprint{5301} & \numprint{4613} & \numprint{11254} & \numprint{5405} \\
& \numprint{16} & \numprint{6001} & \textbf{\numprint{5800}} & \numprint{7263} & \numprint{6949} & \numprint{15725} & \numprint{8243} \\
& \numprint{32} & \numprint{9566} & \textbf{\numprint{9192}} & \numprint{11405} & \numprint{11221} & \numprint{20172} & \numprint{14876} \\
\hline 
\multirow{5}{*}{durbin} & \numprint{2} & \textbf{\numprint{349}} & \textbf{\numprint{349}} & \textbf{\numprint{349}} & \textbf{\numprint{349}} & \numprint{352} & \textbf{\numprint{349}} \\
& \numprint{4} & \numprint{1024} & \textbf{\numprint{1020}} & \numprint{1023} & \textbf{\numprint{1020}} & \numprint{1033} & \textbf{\numprint{1020}} \\
& \numprint{8} & \numprint{2361} & \textbf{\numprint{2339}} & \numprint{2362} & \numprint{2344} & \numprint{2375} & \numprint{2344} \\
& \numprint{16} & \numprint{5030} & \textbf{\numprint{4996}} & \numprint{5027} & \numprint{5000} & \numprint{5047} & \numprint{5018} \\
& \numprint{32} & \numprint{10374} & \numprint{10364} & \numprint{10366} & \textbf{\numprint{10358}} & \numprint{10396} & \numprint{10378} \\
\hline 
\multirow{5}{*}{fdtd-2d} & \numprint{2} & \numprint{2650} & \textbf{\numprint{1756}} & \numprint{3491} & \numprint{3490} & \numprint{3491} & \numprint{3490} \\
& \numprint{4} & \numprint{5549} & \textbf{\numprint{3960}} & \numprint{10473} & \numprint{4294} & \numprint{10474} & \numprint{4269} \\
& \numprint{8} & \numprint{7755} & \textbf{\numprint{6351}} & \numprint{13745} & \numprint{8673} & \numprint{24366} & \numprint{8120} \\
& \numprint{16} & \numprint{10971} & \textbf{\numprint{8959}} & \numprint{19112} & \numprint{13681} & \numprint{34855} & \numprint{15108} \\
& \numprint{32} & \numprint{14110} & \textbf{\numprint{12759}} & \numprint{24248} & \numprint{18726} & \numprint{42703} & \numprint{22036} \\
\hline 
\multirow{5}{*}{gemm} & \numprint{2} & \textbf{\numprint{4200}} & \textbf{\numprint{4200}} & \numprint{6179} & \numprint{4758} & \numprint{5989} & *\numprint{4506} \\
& \numprint{4} & \textbf{\numprint{12600}} & \textbf{\numprint{12600}} & \numprint{18908} & \numprint{14781} & \numprint{18579} & \numprint{14581} \\
& \numprint{8} & \numprint{20931} & \textbf{\numprint{19714}} & \numprint{39528} & \numprint{39290}  & \numprint{41055} & \numprint{35135} \\
& \numprint{16} & \numprint{33978} & \textbf{\numprint{31355}} & \numprint{63139} & *\numprint{63139} & \numprint{77882} & \numprint{76501} \\
& \numprint{32} & \numprint{52721} & \textbf{\numprint{50300}} & \numprint{89660} & *\numprint{89660} & \numprint{117319} & \numprint{115717} \\
\hline 
\multirow{5}{*}{gemver} & \numprint{2} & \numprint{2577} & \textbf{\numprint{480}} & \numprint{1824} & \textbf{\numprint{480}} & \numprint{4800} & \numprint{2947} \\
& \numprint{4} & \numprint{5341} & \textbf{\numprint{2070}} & \numprint{6705} & \numprint{4576} & \numprint{8081} & \numprint{5851} \\
& \numprint{8} & \numprint{10615} & \textbf{\numprint{8305}} & \numprint{10522} & \numprint{8357} & \numprint{15511} & \numprint{9673} \\
& \numprint{16} & \numprint{13432} & \textbf{\numprint{12474}} & \numprint{13263} & \numprint{12618} & \numprint{20260} & \numprint{14005} \\
& \numprint{32} & \numprint{17250} & \numprint{16576} & \numprint{16823} & \textbf{\numprint{16362}} & \numprint{25050} & \numprint{21086} \\
\hline 
\multirow{5}{*}{gesummv} & \numprint{2} & \numprint{350} & \textbf{\numprint{250}} & \numprint{518} & \numprint{501} & \numprint{523} & \numprint{500} \\
& \numprint{4} & \numprint{975} & \textbf{\numprint{750}} & \numprint{927} & \numprint{760} & \numprint{1191} & \numprint{1051} \\
& \numprint{8} & \numprint{1394} & \textbf{\numprint{1250}} & \numprint{1539} & \numprint{1515} & \numprint{2128} & \numprint{2053} \\
& \numprint{16} & \numprint{2247} & \textbf{\numprint{2246}} & \numprint{2600} & \numprint{2582} & \numprint{5403} & \numprint{2971} \\
& \numprint{32} & \numprint{3526} & \textbf{\numprint{3428}} & \numprint{3644} & \numprint{3454} & \numprint{4689} & \numprint{4295} \\
\hline 
\multirow{5}{*}{heat-3d} & \numprint{2} & \textbf{\numprint{1280}} & \textbf{\numprint{1280}} & \numprint{1347} & \textbf{\numprint{1280}} & \numprint{1358} & \textbf{\numprint{1280}} \\
& \numprint{4} & \numprint{3843} & \textbf{\numprint{3840}} & \numprint{3947} & \textbf{\numprint{3840}} & \numprint{4190} & \textbf{\numprint{3840}} \\
& \numprint{8} & \numprint{9427} & \textbf{\numprint{8777}} & \numprint{9222} & \numprint{8960} & \numprint{9776} & \numprint{8960} \\
& \numprint{16} & \numprint{15406} & \numprint{14509} & \numprint{16496} & \numprint{14325} & \numprint{19799} & \textbf{\numprint{14313}} \\
& \numprint{32} & \numprint{21102} & \textbf{\numprint{19382}} & \numprint{22727} & \numprint{20483} & \numprint{28957} & \numprint{21080} \\
\hline
\end{tabular}}
\end{center}
\end{table}
\vfill
\pagebreak
\begin{table}[h!]
\begin{center}
\scriptsize
\vspace*{-2cm}
\caption{Detailed per instance results on the PolyBench benchmark set \cite{PolyBench}. \emph{mlDHGP} refers to our algorithm with \emph{KaHyPar} as undirected hypergraph partitioner for initial partitioning.  \emph{memDHGP} refers to our memetic algorithm that uses mlDHGP equipped with \emph{KaHyPar} as undirected hypergraph partitioner for initial partitioning to build an initial population.  The \emph{Best} column reports the best edge cut found during 8 hours of individual runs. Instances that took longer than 8 hours to compute are marked with a star. For \emph{mlDHGP}, the \emph{Average} column reports the average edge cut of 5 individual runs.  For \emph{memDHGP}, the \emph{Best} column reports the best result found after running for 8 hours.  In general, lower is better.}
\resizebox{.5\columnwidth}{!}{ 
\begin{tabular}{|lr|rr|rr|rr|}
\hline 
& & mlDHGP & memDHGP & \multicolumn{2}{c}{TopoRB} & \multicolumn{2}{|c|}{TopoKWay} \\
\hline 
\multicolumn{1}{|c}{Hypergraph} & \multicolumn{1}{c|}{K} & \multicolumn{1}{r}{Average} & \multicolumn{1}{r|}{Best (8h)} & \multicolumn{1}{r}{Average} & \multicolumn{1}{r|}{Best (8h)} &\multicolumn{1}{r}{Average} & \multicolumn{1}{r|}{Best (8h)}\\
\hline 
\multirow{5}{*}{jacobi-1d} & \numprint{2} & \numprint{401} & \textbf{\numprint{400}} & \numprint{412} & \numprint{402} & \numprint{411} & \numprint{402} \\
& \numprint{4} & \numprint{926} & \textbf{\numprint{793}} & \numprint{1245} & \numprint{1206} & \numprint{1279} & \numprint{1206} \\
& \numprint{8} & \numprint{1587} & \textbf{\numprint{1467}} & \numprint{2900} & \numprint{2814} & \numprint{3053} & \numprint{2793} \\
& \numprint{16} & \numprint{2634} & \textbf{\numprint{2423}} & \numprint{6213} & \numprint{3900} & \numprint{6676} & \numprint{3349} \\
& \numprint{32} & \numprint{3992} & \textbf{\numprint{3786}} & \numprint{8788} & \numprint{5540} & \numprint{13680} & \numprint{4753} \\
\hline 
\multirow{5}{*}{jacobi-2d} & \numprint{2} & \textbf{\numprint{1008}} & \textbf{\numprint{1008}} & \numprint{1053} & \textbf{\numprint{1008}} & \numprint{1049} & \textbf{\numprint{1008}} \\
& \numprint{4} & \numprint{3524} & \textbf{\numprint{2981}} & \numprint{3093} & \numprint{3024} & \numprint{3129} & \numprint{3024} \\
& \numprint{8} & \numprint{5786} & \textbf{\numprint{4995}} & \numprint{7184} & \numprint{6978} & \numprint{7419} & \numprint{6837} \\
& \numprint{16} & \numprint{8198} & \textbf{\numprint{7215}} & \numprint{13070} & \numprint{9282} & \numprint{15715} & \numprint{8992} \\
& \numprint{32} & \numprint{11312} & \textbf{\numprint{10326}} & \numprint{16921} & \numprint{14002} & \numprint{24070} & \numprint{11587} \\

\hline 
\multirow{5}{*}{lu} & \numprint{2} & \numprint{3327} & \numprint{3221} & \numprint{2966} & \textbf{\numprint{2776}} & \numprint{3644} & \numprint{3190} \\
& \numprint{4} & \numprint{5922} & \textbf{\numprint{5735}} & \numprint{6219} & \numprint{5898} & \numprint{9635} & \numprint{9181} \\
& \numprint{8} & \numprint{10218} & \textbf{\numprint{9831}} & \numprint{10971} & \numprint{10837} & \numprint{20432} & \numprint{18592} \\
& \numprint{16} & \numprint{15319} & \textbf{\numprint{15145}} & \numprint{15735} & \numprint{15152} & \numprint{27899} & \numprint{27673} \\
& \numprint{32} & \numprint{22034} & \textbf{\numprint{21652}} & \numprint{23252} & \numprint{22984} & \numprint{36568} & \numprint{36178} \\
\hline 
\multirow{5}{*}{ludcmp} & \numprint{2} & \numprint{2952} & \textbf{\numprint{2887}} & \numprint{3020} & \numprint{2917} & \numprint{4364} & \numprint{3878} \\
& \numprint{4} & \numprint{7546} & \textbf{\numprint{7468}} & \numprint{7631} & \numprint{7479} & \numprint{11193} & \numprint{10758} \\
& \numprint{8} & \numprint{12568} & \numprint{12494} & \numprint{12557} & \textbf{\numprint{12322}} & \numprint{22516} & \numprint{22189} \\
& \numprint{16} & \numprint{18211} & \textbf{\numprint{17933}} & \numprint{20093} & \numprint{19412} & \numprint{31115} & \numprint{30422} \\
& \numprint{32} & \numprint{25273} & \textbf{\numprint{24491}} & \numprint{26447} & \numprint{26164} & \numprint{42154} & \numprint{42154} \\
\hline 
\multirow{5}{*}{mvt} & \numprint{2} & \numprint{446} & \textbf{\numprint{404}} & \numprint{3247} & \numprint{558} & \numprint{11174} & \numprint{468} \\
& \numprint{4} & \numprint{1069} & \textbf{\numprint{818}} & \numprint{2988} & \numprint{1664} & \numprint{14887} & \numprint{8545} \\
& \numprint{8} & \numprint{2425} & \textbf{\numprint{1648}} & \numprint{6860} & \numprint{4187} & \numprint{20852} & \numprint{14909} \\
& \numprint{16} & \numprint{2851} & \textbf{\numprint{2586}} & \numprint{13203} & \numprint{10041} & \numprint{32053} & \numprint{23345} \\
& \numprint{32} & \numprint{6288} & \textbf{\numprint{4295}} & \numprint{14809} & \numprint{10009} & \numprint{40295} & \numprint{34690} \\
\hline 
\multirow{5}{*}{seidel-2d} & \numprint{2} & \textbf{\numprint{838}} & \textbf{\numprint{838}} & \numprint{996} & \numprint{935} & \numprint{1275} & \numprint{938} \\
& \numprint{4} & \numprint{2582} & \textbf{\numprint{2473}} & \numprint{2775} & \numprint{2672} & \numprint{3349} & \numprint{2784} \\
& \numprint{8} & \numprint{4668} & \textbf{\numprint{4274}} & \numprint{6020} & \numprint{5403} & \numprint{7265} & \numprint{4905} \\
& \numprint{16} & \numprint{7247} & \textbf{\numprint{6580}} & \numprint{10166} & \numprint{9045} & \numprint{14873} & \numprint{9157} \\
& \numprint{32} & \numprint{10869} & \textbf{\numprint{9966}} & \numprint{15649} & \numprint{13240} & \numprint{26383} & \numprint{15662} \\
\hline 
\multirow{5}{*}{symm} & \numprint{2} & \numprint{836} & \textbf{\numprint{820}} & \numprint{2915} & \numprint{2346} & \numprint{2946} & \numprint{2808} \\
& \numprint{4} & \numprint{2630} & \textbf{\numprint{2540}} & \numprint{4963} & \numprint{4370} & \numprint{7034} & \numprint{6031} \\
& \numprint{8} & \numprint{6257} & \textbf{\numprint{6107}} & \numprint{9023} & \numprint{8862} & \numprint{11819} & \numprint{9618} \\
& \numprint{16} & \numprint{10721} & \textbf{\numprint{10445}} & \numprint{13520} & \numprint{13251} & \numprint{20199} & \numprint{19794} \\
& \numprint{32} & \numprint{15672} & \textbf{\numprint{15282}} & \numprint{18851} & \numprint{18594} & \numprint{33848} & \numprint{32173} \\
\hline 
\multirow{5}{*}{syr2k} & \numprint{2} & \numprint{900} & \textbf{\numprint{880}} & \numprint{1139} & \numprint{900} & \numprint{1356} & \numprint{900} \\
& \numprint{4} & \numprint{1938} & \textbf{\numprint{1820}} & \numprint{2327} & \numprint{1978} & \numprint{3062} & \numprint{1994} \\
& \numprint{8} & \numprint{3834} & \numprint{3372} & \numprint{3913} & \textbf{\numprint{3198}} & \numprint{6010} & \numprint{3763} \\
& \numprint{16} & \numprint{5579} & \textbf{\numprint{4967}} & \numprint{5868} & \numprint{5294} & \numprint{10551} & \numprint{5354} \\
& \numprint{32} & \numprint{7912} & \numprint{7590} & \numprint{8621} & \textbf{\numprint{7520}} & \numprint{16163} & \numprint{9684} \\
\hline 
\multirow{5}{*}{syrk} & \numprint{2} & \textbf{\numprint{3240}} & \textbf{\numprint{3240}} & \numprint{3393} & \numprint{3376} & \numprint{3854} & \textbf{\numprint{3240}} \\
& \numprint{4} & \numprint{7390} & \textbf{\numprint{7320}} & \numprint{10083} & \numprint{10079} & \numprint{10431} & \numprint{9482} \\
& \numprint{8} & \numprint{13566} & \textbf{\numprint{13202}} & \numprint{13924} & \numprint{13924} & \numprint{19379} & \numprint{17118} \\
& \numprint{16} & \numprint{20121} & \textbf{\numprint{19674}} & \numprint{31052} & *\numprint{30851} & \numprint{30102} & \numprint{28704} \\
& \numprint{32} & \numprint{28222} & \textbf{\numprint{27446}} & \numprint{42805} & *\numprint{42805} & \numprint{47622} & \numprint{46759} \\
\hline 
\multirow{5}{*}{trisolv} & \numprint{2} & \textbf{\numprint{279}} & \textbf{\numprint{279}} & \numprint{280} & \textbf{\numprint{279}} & \numprint{283} & \numprint{280} \\
& \numprint{4} & \numprint{620} & \numprint{595} & \numprint{777} & \numprint{600} & \numprint{643} & \textbf{\numprint{581}} \\
& \numprint{8} & \numprint{1088} & \textbf{\numprint{1054}} & \numprint{1260} & \numprint{1133} & \numprint{1366} & \numprint{1289} \\
& \numprint{16} & \numprint{1788} & \textbf{\numprint{1742}} & \numprint{2008} & \numprint{1808} & \numprint{2622} & \numprint{2420} \\
& \numprint{32} & \numprint{2783} & \textbf{\numprint{2683}} & \numprint{3347} & \numprint{3020} & \numprint{4420} & \numprint{3984} \\
\hline 
\multirow{5}{*}{trmm} & \numprint{2} & \numprint{2704} & \textbf{\numprint{1844}} & \numprint{3755} & \numprint{3579} & \numprint{4113} & \numprint{3440} \\
& \numprint{4} & \numprint{6226} & \textbf{\numprint{5673}} & \numprint{7311} & \numprint{7167} & \numprint{11452} & \numprint{8793} \\
& \numprint{8} & \numprint{10082} & \textbf{\numprint{9914}} & \numprint{13669} & \numprint{13484} & \numprint{19559} & \numprint{12482} \\
& \numprint{16} & \numprint{16173} & \textbf{\numprint{15472}} & \numprint{20933} & \numprint{20933} & \numprint{30026} & \numprint{20348} \\
& \numprint{32} & \numprint{22126} & \textbf{\numprint{21437}} & \numprint{27168} & \numprint{27168} & \numprint{42503} & \numprint{41895} \\
\hline
\hline
\textbf{Mean} & & \numprint{4447} & \numprint{3900} & \numprint{5698} & \numprint{4853} & \numprint{8247} & \numprint{6045} \\ 
\hline
\end{tabular}}
\end{center}
\end{table}
\begin{table}[h!]
\begin{center}
\scriptsize
\vspace*{-2cm}
\caption{Detailed per instance results on the ISPD98 benchmark suite \cite{ISPD98}. \emph{mlDHGP} refers to our algorithm with \emph{KaHyPar} as undirected hypergraph partitioner for initial partitioning.  \emph{memDHGP} refers to our memetic algorithm that uses mlDHGP equipped with \emph{KaHyPar} as undirected hypergraph partitioner for initial partitioning to build an initial population.  The \emph{Best} column reports the best edge cut found during 8 hours of individual runs. For \emph{mlDHGP}, the \emph{Average} column reports the average edge cut of 5 individual runs.  For \emph{memDHGP}, the \emph{Best} column reports the best result found after running for 8 hours.  In general, lower is better.}

\resizebox{.5\columnwidth}{!}{ 
\begin{tabular}{|lr|rr|rr|rr|}
\hline 
& & mlDHGP & memDHGP & \multicolumn{2}{c}{TopoRB} & \multicolumn{2}{|c|}{TopoKWay} \\
\hline 
\multicolumn{1}{|c}{Hypergraph} & \multicolumn{1}{c|}{K} & \multicolumn{1}{r}{Average} & \multicolumn{1}{r|}{Best (8h)} & \multicolumn{1}{r}{Average} & \multicolumn{1}{r|}{Best (8h)} &\multicolumn{1}{r}{Average} & \multicolumn{1}{r|}{Best (8h)}\\
\hline 

\multirow{5}{*}{ibm01} & \numprint{2} & \numprint{838} & \textbf{\numprint{629}} & \numprint{877} & \numprint{659} & \numprint{1267} & \numprint{660} \\
& \numprint{4} & \numprint{1835} & \textbf{\numprint{1427}} & \numprint{2035} & \numprint{1684} & \numprint{4921} & \numprint{2615} \\
& \numprint{8} & \numprint{2923} & \textbf{\numprint{2136}} & \numprint{3512} & \numprint{2670} & \numprint{6513} & \numprint{4153} \\
& \numprint{16} & \numprint{3764} & \textbf{\numprint{3049}} & \numprint{4584} & \numprint{3710} & \numprint{8271} & \numprint{6032} \\
& \numprint{32} & \numprint{4774} & \textbf{\numprint{4013}} & \numprint{5626} & \numprint{4706} & \numprint{9894} & \numprint{6652} \\
\hline 
\multirow{5}{*}{ibm02} & \numprint{2} & \numprint{2222} & \textbf{\numprint{1869}} & \numprint{2629} & \numprint{1990} & \numprint{3048} & \numprint{2319} \\
& \numprint{4} & \numprint{4391} & \textbf{\numprint{3247}} & \numprint{5296} & \numprint{4017} & \numprint{7520} & \numprint{5185} \\
& \numprint{8} & \numprint{6898} & \textbf{\numprint{5674}} & \numprint{8561} & \numprint{6677} & \numprint{11208} & \numprint{9485} \\
& \numprint{16} & \numprint{9787} & \textbf{\numprint{8481}} & \numprint{10678} & \numprint{9300} & \numprint{14195} & \numprint{12709} \\
& \numprint{32} & \numprint{12545} & \textbf{\numprint{11448}} & \numprint{13773} & \numprint{12362} & \numprint{18141} & \numprint{14596} \\
\hline 
\multirow{5}{*}{ibm03} & \numprint{2} & \numprint{3782} & \textbf{\numprint{2242}} & \numprint{3772} & \numprint{2862} & \numprint{4306} & \numprint{2932} \\
& \numprint{4} & \numprint{5955} & \textbf{\numprint{4231}} & \numprint{6335} & \numprint{4748} & \numprint{8661} & \numprint{6746} \\
& \numprint{8} & \numprint{7679} & \textbf{\numprint{5911}} & \numprint{8478} & \numprint{6771} & \numprint{12510} & \numprint{10131} \\
& \numprint{16} & \numprint{9179} & \textbf{\numprint{7386}} & \numprint{10278} & \numprint{8601} & \numprint{15725} & \numprint{12304} \\
& \numprint{32} & \numprint{10051} & \textbf{\numprint{8496}} & \numprint{12271} & \numprint{10116} & \numprint{18507} & \numprint{14162} \\
\hline 
\multirow{5}{*}{ibm04} & \numprint{2} & \numprint{3080} & \textbf{\numprint{717}} & \numprint{4448} & \numprint{3044} & \numprint{5252} & \numprint{3204} \\
& \numprint{4} & \numprint{5232} & \textbf{\numprint{2467}} & \numprint{6175} & \numprint{3707} & \numprint{9871} & \numprint{6086} \\
& \numprint{8} & \numprint{7239} & \textbf{\numprint{5339}} & \numprint{9919} & \numprint{7304} & \numprint{13859} & \numprint{9917} \\
& \numprint{16} & \numprint{9415} & \textbf{\numprint{7343}} & \numprint{11868} & \numprint{10029} & \numprint{17680} & \numprint{12584} \\
& \numprint{32} & \numprint{11129} & \textbf{\numprint{9259}} & \numprint{13795} & \numprint{11947} & \numprint{21342} & \numprint{16273} \\
\hline 
\multirow{5}{*}{ibm05} & \numprint{2} & \numprint{4630} & \textbf{\numprint{3954}} & \numprint{4799} & \numprint{4232} & \numprint{4952} & \numprint{4248} \\
& \numprint{4} & \numprint{7629} & \textbf{\numprint{5930}} & \numprint{9574} & \numprint{7222} & \numprint{11693} & \numprint{7933} \\
& \numprint{8} & \numprint{10434} & \textbf{\numprint{8612}} & \numprint{13292} & \numprint{10339} & \numprint{17575} & \numprint{11821} \\
& \numprint{16} & \numprint{13095} & \textbf{\numprint{11285}} & \numprint{16394} & \numprint{13566} & \numprint{21884} & \numprint{16613} \\
& \numprint{32} & \numprint{15371} & \textbf{\numprint{13837}} & \numprint{18577} & \numprint{15938} & \numprint{25750} & \numprint{20079} \\
\hline 
\multirow{5}{*}{ibm06} & \numprint{2} & \numprint{4486} & \textbf{\numprint{2730}} & \numprint{5624} & \numprint{3839} & \numprint{7027} & \numprint{4279} \\
& \numprint{4} & \numprint{8189} & \textbf{\numprint{5858}} & \numprint{8789} & \numprint{6648} & \numprint{14557} & \numprint{11971} \\
& \numprint{8} & \numprint{10203} & \textbf{\numprint{8281}} & \numprint{11483} & \numprint{9590} & \numprint{19122} & \numprint{15012} \\
& \numprint{16} & \numprint{12720} & \textbf{\numprint{10157}} & \numprint{14123} & \numprint{11751} & \numprint{23880} & \numprint{20591} \\
& \numprint{32} & \numprint{15155} & \textbf{\numprint{12179}} & \numprint{17588} & \numprint{14851} & \numprint{30019} & \numprint{24375} \\
\hline 
\multirow{5}{*}{ibm07} & \numprint{2} & \numprint{5355} & \textbf{\numprint{3680}} & \numprint{8273} & \numprint{3871} & \numprint{8831} & \numprint{4602} \\
& \numprint{4} & \numprint{10343} & \textbf{\numprint{6250}} & \numprint{11463} & \numprint{7130} & \numprint{16176} & \numprint{11691} \\
& \numprint{8} & \numprint{12386} & \textbf{\numprint{8993}} & \numprint{13861} & \numprint{9482} & \numprint{21580} & \numprint{16862} \\
& \numprint{16} & \numprint{13927} & \textbf{\numprint{11870}} & \numprint{17289} & \numprint{14109} & \numprint{27718} & \numprint{23060} \\
& \numprint{32} & \numprint{16880} & \textbf{\numprint{14264}} & \numprint{20418} & \numprint{17595} & \numprint{34122} & \numprint{26761} \\
\hline 
\multirow{5}{*}{ibm08} & \numprint{2} & \numprint{12865} & \numprint{9344} & \numprint{11639} & \textbf{\numprint{9331}} & \numprint{12536} & \numprint{11281} \\
& \numprint{4} & \numprint{18373} & \numprint{16860} & \numprint{18385} & \textbf{\numprint{15454}} & \numprint{22071} & \numprint{20504} \\
& \numprint{8} & \numprint{22238} & \numprint{20526} & \numprint{21703} & \textbf{\numprint{20013}} & \numprint{28396} & \numprint{24835} \\
& \numprint{16} & \numprint{25572} & \textbf{\numprint{23100}} & \numprint{27093} & \numprint{25072} & \numprint{34302} & \numprint{30639} \\
& \numprint{32} & \numprint{29667} & \textbf{\numprint{27425}} & \numprint{31907} & \numprint{29150} & \numprint{40661} & \numprint{34423} \\
\hline 
\multirow{5}{*}{ibm09} & \numprint{2} & \numprint{5593} & \textbf{\numprint{3357}} & \numprint{11804} & \numprint{9159} & \numprint{13130} & \numprint{9651} \\
& \numprint{4} & \numprint{10610} & \textbf{\numprint{6416}} & \numprint{18274} & \numprint{13852} & \numprint{21028} & \numprint{17693} \\
& \numprint{8} & \numprint{12053} & \textbf{\numprint{8726}} & \numprint{20346} & \numprint{16942} & \numprint{26819} & \numprint{21107} \\
& \numprint{16} & \numprint{14987} & \textbf{\numprint{11588}} & \numprint{22256} & \numprint{20090} & \numprint{33138} & \numprint{28843} \\
& \numprint{32} & \numprint{17802} & \textbf{\numprint{14449}} & \numprint{25313} & \numprint{22858} & \numprint{39478} & \numprint{34857} \\
\hline 
\end{tabular}}
\end{center}

\end{table}
\begin{table}[h!]
\begin{center}
\scriptsize
\vspace*{-2cm}
\caption{Detailed per instance results on the ISPD98 benchmark suite \cite{ISPD98}. \emph{mlDHGP} refers to our algorithm with \emph{KaHyPar} as undirected hypergraph partitioner for initial partitioning.  \emph{memDHGP} refers to our memetic algorithm that uses mlDHGP equipped with \emph{KaHyPar} as undirected hypergraph partitioner for initial partitioning to build an initial population.  The \emph{Best} column reports the best edge cut found during 8 hours of individual runs. For \emph{mlDHGP}, the \emph{Average} column reports the average edge cut of 5 individual runs.  For \emph{memDHGP}, the \emph{Best} column reports the best result found after running for 8 hours.  In general, lower is better.}
\label{tab:detailedresultshg:ispd98:2}

\resizebox{.5\columnwidth}{!}{ 
\begin{tabular}{|lr|rr|rr|rr|}
\hline 
& & mlDHGP & memDHGP & \multicolumn{2}{c}{TopoRB} & \multicolumn{2}{|c|}{TopoKWay} \\
\hline 
\multicolumn{1}{|c}{Hypergraph} & \multicolumn{1}{c|}{K} & \multicolumn{1}{r}{Average} & \multicolumn{1}{r|}{Best (8h)} & \multicolumn{1}{r}{Average} & \multicolumn{1}{r|}{Best (8h)} &\multicolumn{1}{r}{Average} & \multicolumn{1}{r|}{Best (8h)}\\
\hline 
\multirow{5}{*}{ibm10} & \numprint{2} & \numprint{12129} & \textbf{\numprint{8288}} & \numprint{11717} & \numprint{8315} & \numprint{17445} & \numprint{11372} \\
& \numprint{4} & \numprint{18728} & \textbf{\numprint{11809}} & \numprint{15691} & \numprint{12922} & \numprint{24176} & \numprint{18224} \\
& \numprint{8} & \numprint{22141} & \textbf{\numprint{16401}} & \numprint{22146} & \numprint{18837} & \numprint{36282} & \numprint{28691} \\
& \numprint{16} & \numprint{24929} & \textbf{\numprint{20199}} & \numprint{27168} & \numprint{24613} & \numprint{49440} & \numprint{39218} \\
& \numprint{32} & \numprint{30100} & \textbf{\numprint{26238}} & \numprint{33799} & \numprint{30462} & \numprint{59161} & \numprint{51361} \\
\hline 
\multirow{5}{*}{ibm11} & \numprint{2} & \numprint{10669} & \textbf{\numprint{6550}} & \numprint{10943} & \numprint{7618} & \numprint{13320} & \numprint{8190} \\
& \numprint{4} & \numprint{16257} & \textbf{\numprint{10447}} & \numprint{18761} & \numprint{15073} & \numprint{25457} & \numprint{20894} \\
& \numprint{8} & \numprint{17992} & \textbf{\numprint{12341}} & \numprint{23071} & \numprint{18430} & \numprint{36747} & \numprint{29617} \\
& \numprint{16} & \numprint{20197} & \textbf{\numprint{16198}} & \numprint{27404} & \numprint{23673} & \numprint{42915} & \numprint{35907} \\
& \numprint{32} & \numprint{23409} & \textbf{\numprint{19767}} & \numprint{30592} & \numprint{27495} & \numprint{50761} & \numprint{43688} \\
\hline 
\multirow{5}{*}{ibm12} & \numprint{2} & \numprint{15449} & \numprint{11349} & \numprint{14881} & \textbf{\numprint{10588}} & \numprint{14858} & \numprint{12725} \\
& \numprint{4} & \numprint{20307} & \textbf{\numprint{15652}} & \numprint{20215} & \numprint{16538} & \numprint{23398} & \numprint{18847} \\
& \numprint{8} & \numprint{23036} & \textbf{\numprint{18126}} & \numprint{24916} & \numprint{21053} & \numprint{36501} & \numprint{31499} \\
& \numprint{16} & \numprint{28437} & \textbf{\numprint{23367}} & \numprint{30176} & \numprint{26434} & \numprint{48537} & \numprint{35157} \\
& \numprint{32} & \numprint{34536} & \textbf{\numprint{27911}} & \numprint{38183} & \numprint{33930} & \numprint{62718} & \numprint{53703} \\
\hline

\multirow{5}{*}{ibm13} & \numprint{2} & \numprint{11893} & \numprint{8695} & \numprint{12790} & \textbf{\numprint{8284}} & \numprint{19262} & \numprint{10593} \\
& \numprint{4} & \numprint{14791} & \textbf{\numprint{10285}} & \numprint{21166} & \numprint{12883} & \numprint{43564} & \numprint{34438} \\
& \numprint{8} & \numprint{21405} & \textbf{\numprint{14330}} & \numprint{32543} & \numprint{21452} & \numprint{58298} & \numprint{48381} \\
& \numprint{16} & \numprint{25313} & \textbf{\numprint{16761}} & \numprint{35524} & \numprint{28945} & \numprint{70447} & \numprint{58881} \\
& \numprint{32} & \numprint{29676} & \textbf{\numprint{26017}} & \numprint{43126} & \numprint{37405} & \numprint{92014} & \numprint{81315} \\
\hline 
\multirow{5}{*}{ibm14} & \numprint{2} & \numprint{24379} & \numprint{14713} & \numprint{15305} & \textbf{\numprint{14219}} & \numprint{21228} & \numprint{18308} \\
& \numprint{4} & \numprint{30912} & \numprint{21613} & \numprint{24657} & \textbf{\numprint{21539}} & \numprint{38520} & \numprint{33361} \\
& \numprint{8} & \numprint{36370} & \textbf{\numprint{30710}} & \numprint{36889} & \numprint{32478} & \numprint{49762} & \numprint{44724} \\
& \numprint{16} & \numprint{42321} & \textbf{\numprint{35598}} & \numprint{44721} & \numprint{40298} & \numprint{63893} & \numprint{58170} \\
& \numprint{32} & \numprint{48741} & \textbf{\numprint{43979}} & \numprint{53609} & \numprint{48580} & \numprint{78351} & \numprint{72065} \\
\hline 
\multirow{5}{*}{ibm15} & \numprint{2} & \numprint{27810} & \textbf{\numprint{19804}} & \numprint{28396} & \numprint{24432} & \numprint{38193} & \numprint{29247} \\
& \numprint{4} & \numprint{44069} & \textbf{\numprint{33151}} & \numprint{52804} & \numprint{46517} & \numprint{79810} & \numprint{74511} \\
& \numprint{8} & \numprint{51886} & \textbf{\numprint{38306}} & \numprint{65971} & \numprint{57918} & \numprint{102738} & \numprint{93440} \\
& \numprint{16} & \numprint{58961} & \textbf{\numprint{49687}} & \numprint{74815} & \numprint{68546} & \numprint{119898} & \numprint{105971} \\
& \numprint{32} & \numprint{66287} & \textbf{\numprint{56374}} & \numprint{82252} & \numprint{75762} & \numprint{141076} & \numprint{129211} \\
\hline 
\multirow{5}{*}{ibm16} & \numprint{2} & \numprint{25941} & \textbf{\numprint{11494}} & \numprint{26062} & \numprint{19608} & \numprint{34882} & \numprint{30835} \\
& \numprint{4} & \numprint{46933} & \textbf{\numprint{36124}} & \numprint{50521} & \numprint{45089} & \numprint{80648} & \numprint{61926} \\
& \numprint{8} & \numprint{57761} & \textbf{\numprint{45328}} & \numprint{61132} & \numprint{55221} & \numprint{97002} & \numprint{89035} \\
& \numprint{16} & \numprint{67904} & \textbf{\numprint{58471}} & \numprint{72549} & \numprint{64470} & \numprint{114337} & \numprint{105134} \\
& \numprint{32} & \numprint{80591} & \textbf{\numprint{69325}} & \numprint{83631} & \numprint{77873} & \numprint{135371} & \numprint{127061} \\
\hline 
\multirow{5}{*}{ibm17} & \numprint{2} & \numprint{36934} & \numprint{32507} & \numprint{27629} & \textbf{\numprint{25282}} & \numprint{36665} & \numprint{29655} \\
& \numprint{4} & \numprint{47186} & \numprint{39301} & \numprint{42638} & \textbf{\numprint{37519}} & \numprint{64506} & \numprint{53945} \\
& \numprint{8} & \numprint{62896} & \textbf{\numprint{54230}} & \numprint{61043} & \numprint{55706} & \numprint{86830} & \numprint{78413} \\
& \numprint{16} & \numprint{74427} & \textbf{\numprint{65672}} & \numprint{77867} & \numprint{70777} & \numprint{109982} & \numprint{101743} \\
& \numprint{32} & \numprint{86597} & \textbf{\numprint{80152}} & \numprint{93864} & \numprint{87711} & \numprint{129758} & \numprint{121273} \\
\hline 
\multirow{5}{*}{ibm18} & \numprint{2} & \numprint{21296} & \textbf{\numprint{16338}} & \numprint{20830} & \numprint{18825} & \numprint{24096} & \numprint{18409} \\
& \numprint{4} & \numprint{36235} & \textbf{\numprint{28131}} & \numprint{33642} & \numprint{31101} & \numprint{50926} & \numprint{39056} \\
& \numprint{8} & \numprint{49742} & \textbf{\numprint{38947}} & \numprint{45337} & \numprint{41406} & \numprint{73198} & \numprint{65903} \\
& \numprint{16} & \numprint{57312} & \textbf{\numprint{47532}} & \numprint{60672} & \numprint{55461} & \numprint{94165} & \numprint{84746} \\
& \numprint{32} & \numprint{67770} & \textbf{\numprint{58061}} & \numprint{76324} & \numprint{69222} & \numprint{114532} & \numprint{104020} \\
\hline
\hline
\textbf{Mean} & & \numprint{16151} & \numprint{12245} & \numprint{18344} & \numprint{15030} & \numprint{26839} & \numprint{21246} \\ 
\hline
\end{tabular}}
\end{center}
\end{table}
\end{landscape}
\fi
\end{appendix}
\end{document}

%% file: makros.tex
\usepackage{amsfonts}

\newcommand{\realrange}[2]{\left[#1, #2\right]}

\newcommand{\unitrange}[2]{\realrange{0}{1}}

\newcommand{\llabel}[1]{\label{\labelprefix:#1}}
\newcommand{\labelprefix}{} 

\newcommand{\discussionsize}{\small}

\newcommand{\frage}[1]{}

\newenvironment{code}{\noindent
\begin{tabbing}%
\hspace{2em}\=\hspace{2em}\=\hspace{2em}\=\hspace{2em}\=\hspace{2em}\=%
\hspace{2em}\=\hspace{2em}\=\hspace{2em}\=\hspace{2em}\=\hspace{2em}\=%
\kill}{\end{tabbing}}

\newcommand{\labelcommand}{}
\newcommand{\captiontext}{}
\newsavebox{\codeparam}
\newcounter{lineNumber}
\newenvironment{disscodepos}[3]{%
\renewcommand{\labelcommand}{#2}%
\renewcommand{\captiontext}{#3}%
\sbox{\codeparam}{\parbox{\textwidth}{#3}}%
\begin{figure}[#1]\begin{center}\begin{code}\setcounter{lineNumber}{1}}{%
\end{code}\end{center}\caption{\llabel{\labelcommand}\captiontext}\end{figure}}

{\end{disscodepos}}

\newcommand{\While}    {{\bf while\ }}

\newcommand{\For}      {{\bf for\ }}

\newcommand{\If}       {{\bf if\ }}

\newdimen\endofsize\endofsize=0.5em
\def\endofbeweis{~\quad\hglue\hsize minus\hsize
                 \hbox{\vrule height \endofsize width
\endofsize}\par}
